%% file: main.tex
\documentclass[sigconf]{acmart}

\usepackage{graphicx}
\usepackage{balance}  
\usepackage{multirow}

\usepackage{booktabs} 
\usepackage{subfigure}
\usepackage{color}
\usepackage{algorithm}
\usepackage{algorithmicx}
\usepackage{algpseudocode}
\usepackage{bbm}
\usepackage{hyperref}

\AtBeginDocument{%
  \providecommand\BibTeX{{%
    \normalfont B\kern-0.5em{\scshape i\kern-0.25em b}\kern-0.8em\TeX}}}

\newcommand{\pagerank}{{Pagerank}}

\newcommand{\neighborlocal}{{\sc LFPR$_N$}}
\newcommand{\uniformlocal}{{\sc LFPR$_U$}}
\newcommand{\proportionallocal}{{\sc LFPR$_P$}}

\newcommand{\sensitive}{{\sc FSPR}}

\newcommand{\allPR}{\mathcal{PR}}

\newcommand{\residuallocal}{{\sc LFPR$_O$}}

\newcommand{\Ivec}{\mathbf{I}}
\newcommand{\pvec}{\mathbf{p}}
\newcommand{\wvec}{\mathbf{w}}
\newcommand{\vvec}{\mathbf{v}}

\newcommand{\evec}{\mathbf{e}}
\newcommand{\rvec}{\mathbf{r}}
\newcommand{\fvec}{\mathbf{f}}
\newcommand{\Pmatrix}{\mathbf{P}}

\newcommand{\Qmatrix}{\mathbf{Q}}
\newcommand{\Imatrix}{\mathbf{I}}
\newcommand{\Xmatrix}{\mathbf{X}}
\newcommand{\Ymatrix}{\mathbf{Y}}

\newcommand{\PR}{\textsc{PR}}
\newcommand{\OPR}{\textsc{OPR}}
\newcommand{\FPR}{\textsc{FPR}}
\newcommand{\PPR}{\overline{\PR_i}}

\newcommand{\dvec}{\mathbf{\delta}}

\newcommand{\xvec}{\mathbf{x}}
\newcommand{\yvec}{\mathbf{y}}

\newtheorem{problem}{Problem}

\newtheorem{definition}{Definition}
\newtheorem{lemma}{Lemma}
\newtheorem{theorem}{Theorem}

\setcopyright{iw3c2w3}

\begin{document}

\copyrightyear{2021}
\acmYear{2021}
\acmConference[WWW '21]{Proceedings of the Web Conference 2021}{April 19--23, 2021}{Ljubljana, Slovenia}
\acmBooktitle{Proceedings of the Web Conference 2021 (WWW '21), April 19--23, 2021, Ljubljana, Slovenia}
\acmPrice{}
\acmDOI{10.1145/3442381.3450065}
\acmISBN{978-1-4503-8312-7/21/04}

\title{Fairness-Aware PageRank}

\author{Sotiris Tsioutsiouliklis}
\affiliation{ \institution{University of Ioannina}
	\country{Greece}}
\email{stsiouts@cse.uoi.gr}



\author{Evaggelia Pitoura}
\affiliation{ \institution{University of Ioannina}
	\country{Greece}}
\email{pitoura@cse.uoi.gr}

\author{Panayiotis Tsaparas}
\affiliation{ \institution{University of Ioannina}
	\country{Greece}}
\email{tsap@cse.uoi.gr}

\author{Ilias Kleftakis}
\affiliation{ \institution{University of Ioannina}
	\country{Greece}}
\email{ikleftakis@cse.uoi.gr}

\author{Nikolaos Mamoulis}
\affiliation{ \institution{University of Ioannina}
\country{Greece}}
\email{nikos@cse.uoi.gr}

\renewcommand{\shortauthors}{Tsioutsiouliklis and Pitoura, et al.}

\begin{abstract}
Algorithmic fairness has attracted significant attention in the past years. In this paper, we consider fairness for link analysis and in particular for the celebrated Pagerank algorithm. Given that the nodes in a network belong to groups (for example, based on demographic or other characteristics), we provide a parity-based definition of fairness that imposes constraints on the proportion of Pagerank allocated to the members of each group. We propose two families of fair Pagerank algorithms: the first (Fairness-Sensitive Pagerank) modifies the jump vector of the Pagerank algorithm to enforce fairness; the second (Locally Fair Pagerank) imposes a fair behavior per node. We then define a stronger fairness requirement, termed universal personalized fairness, that asks that the derived personalized pageranks of all nodes are fair. We prove that the locally fair algorithms achieve also universal personalized fairness, and furthermore, we prove that this is the only family of algorithms with this property, establishing an equivalence between universal personalized fairness and local fairness. We also consider the problem of achieving fairness while minimizing the utility loss with respect to the original Pagerank algorithm. We present experiments with real and synthetic networks that examine the fairness of the original Pagerank and demonstrate qualitatively and quantitatively the properties of our algorithms.

\end{abstract}



\maketitle

\input{sec1}

\input{sec2}

\input{sec3}

\input{sec4}

\input{sec5}

\input{sec6}

\input{sec7}

\input{sec8}

\begin{acks}
	Research work  supported by the Hellenic Foundation for Research and Innovation (H.F.R.I.)
	 under the “1st Call for H.F.R.I. Research Projects to Support Faculty Members \& Researchers and Procure High-Value Research Equipment” 
	Project Number:  HFRI-FM17-1873, GraphTempo.
	\end{acks}

\bibliographystyle{ACM-Reference-Format}
\bibliography{main}

\end{document}

%% file: sec1.tex
\section{Introduction}
Today, algorithmic systems driven by large amounts of data
are increasingly being used in all aspects of life. 
Often, such systems are being used to assist, or, even replace
human decision-making.
This increased dependence on algorithms has given rise to the field of algorithmic fairness, where the goal is to ensure that algorithms do not exhibit biases towards specific individuals, or groups of users (see e.g., \cite{fairness-study} for a survey).
We also live in a connected world where networks, be it, social, communication,
interaction, or cooperation networks, play a central role.
However, surprisingly,  fairness in networks   
has received less attention.

Link analysis algorithms, such as Pagerank~\cite{pagerank}, take a graph as input and use the structure of the graph to determine the relative importance of its nodes.
The output of the algorithms is a numerical weight for each node that reflects its importance. The weights are
used to produce a ranking of the nodes, but also as  input features in a variety of
machine learning algorithms including classification~\cite{spam-classifier}, and
 search result ranking~\cite{pagerank}.

Pagerank performs a random walk on the input graph, and weights the nodes according to the stationary probability distribution of this walk.
At each step, the random walk restarts with probability $\gamma$, where
the restart node is selected according to a``jump'' distribution vector $\vvec$.
Since its introduction
in the Google search engine, Pagerank has been the
cornerstone algorithm in several applications (see, e.g., \cite{pagerank-survey}).
Previous research on the fairness of centrality measures has considered only degrees and found biases that arise as a network
evolves
\cite{glass-ceiling,glass-ceiling-recommend}, or has studied general notions of fairness in graphs based on the premise that similar nodes should get similar outputs  \cite{inform}.
In this work, we focus on the fairness of the Pagerank algorithm.
 
 As in previous research, we view fairness as  lack of discrimination against a
 protected group  defined by the value of a sensitive attribute, such as, gender, or race \cite{fairness-study}. 
 We operationalize this view by saying that a link analysis algorithm is \textit{$\phi$-fair}, if  the fraction of the total weight allocated to the members of the protected group is $\phi$. 
 The value of $\phi$ is a parameter that can be used to implement different fairness policies. For example, by setting $\phi$  equal to the ratio of the protected nodes in the graph,  we ask that the protected nodes have a share in the weights proportional to their share in the population, a property also known as demographic parity
 \cite{fairness-awarness}. 
We also consider \emph{targeted} fairness, where we focus on a specific subset of nodes to which we want to allocate weights in a fair manner.

We revisit  Pagerank  through the lens of  our fairness definitions,
and we consider the problem of defining families of Pagerank algorithms that are fair.
We also define the \emph{utility loss} of a fair algorithm as the difference between its output and the output of the original Pagerank algorithm, and we pose the problem of achieving fairness while minimizing utility.

We consider two approaches for achieving fairness. 
The first family of algorithms we consider is  the \emph{fairness-sensitive} Pagerank family which exploits the jump vector $\vvec$. 
There has been a lot of work on modifying the jump vector to obtain variants of Pagerank biased towards a specific set of nodes. 
The topic-sensitive Pagerank algorithm \cite{topic-sensitive} is such an example, where the probability is assigned to nodes of a specific topic. 
In this paper, we take the novel approach of using the jump vector to achieve $\phi$-fairness.
We determine the conditions under which this is feasible 
and formulate  the problem of finding the jump vector that achieves $\phi$-fairness while minimizing utility loss 
as a convex optimization problem. 

Our second family of algorithms takes a microscopic view of fairness by looking at the behavior of each individual node in the graph.
Implicitly, a link analysis algorithm assumes that links in the graph correspond to endorsements between the nodes. Therefore, we can view each node, as an agent that \emph{endorses} (or \emph{votes for}) the nodes that it links to. Pagerank defines a process that takes these individual actions of the nodes and transforms them into a global weighting of the nodes.
We thus introduce, the \textit{locally fair PageRank algorithms}, where each individual node acts fairly by distributing its own pagerank to the protected and non-protected groups according to the fairness ratio $\phi$.
Local fairness defines a dynamic process that can be viewed as a  \textit{fair random walk}, where \emph{at each step} of the random walk (not only at convergence), the  probability of being at a node of the protected group is $\phi$.
 
In our first locally fair PageRank algorithm, termed the \textit{neighborhood locally fair} Pagerank algorithm,  each node distributes its pagerank fairly among its immediate neighbors, allocating a fraction $\phi$ to the neighbors in the protected group, and $1-\phi$ to the neighbors in the non-protected group.
The \textit{residual-based locally fair} Pagerank algorithms generalizes this idea.
Consider a node $i$ that has less neighbors in the protected group than $\phi$. The node distributes an equal portion of its pagerank to each of its neighbors and a residual portion $\delta(i)$
to members in the protected group but not necessarily in its own neighborhood.
The residual is allocated based on \textit{a residual redistribution policy}, which allows us to control the fairness policy. 
In this paper, we exploit a residual redistribution policy that minimizes the utility loss.

We then define a stronger fairness requirement, termed universal personalized fairness, that asks that the derived personalized pageranks of all nodes are fair. We prove that the locally fair algorithms achieve also universal personalized fairness. 
Surprisingly, the locally fair algorithms are the \emph{only} family of algorithms with this property. Thus, we show that an algorithm is locally fair, if and only if, it is universally personalized fair.

We use real and synthetic datasets to study the conditions under which   {\pagerank} and personalized {\pagerank} are fair.
We also evaluate both quantitatively and qualitatively  the output of our fairness-aware algorithms.

In summary, in this paper, we  make the following contributions:
\begin{itemize}
	\item We initiate a study of fairness for the {\pagerank} algorithm. 
	\item We propose the fairness-sensitive  Pagerank family that modifies the jump vector so as to achieve fairness, and the locally fair Pagerank family that guarantees that individually each node behaves in a fair manner. 
\item 	We prove that local fairness implies universal personalized fairness and also that this is the only family of algorithms with this property, establishing an equivalence between local fairness and universal personalized fairness.
	\item We perform experiments on several datasets to study the conditions under which Pagerank unfairness emerges and evaluate the utility loss for enforcing fairness. 

\end{itemize}

The remainder of this paper is structured as follows.
In Section \ref{sec:definitions}, we provide definitions of fairness and we formulate our problems. In Sections \ref{sec:fairness-sensitive} and
\ref{sec:local-fair}, we introduce the fairness sensitive and 
the locally fair families of {\pagerank} algorithms.
In Section \ref{sec:universal}, we discuss personalized fairness and we show an equivalence between local and universal personalized fairness. The results of our experimental evaluation are presented in Section \ref{sec:experiments}. Finally, we present related research in Section \ref{sec:related-work} and our conclusions in Section \ref{sec:conclusions}.

%% file: sec2.tex
\section{Definitions}
\label{sec:definitions}
In this section, we first present background material and then we define Pagerank fairness.

\subsection{Preliminaries}
The {\pagerank} algorithm~\cite{pagerank} pioneered link analysis for weighting and ranking the nodes of a graph. It was popularized by its application in the Google search engine, but it has found a wide range of applications in different settings~\cite{pagerank-survey}. The algorithm takes as input a graph $G = (V,E)$, and produces a scoring vector, that assigns a weight to each node $v \in V$ in the graph. The scoring vector is the stationary distribution of a random walk on the graph $G$. 

The {\pagerank} random walk is a random walk with restarts. It is parameterized by the value $\gamma$, which is the probability that the random walk will restart at any step.
The node from which the random walk restarts is selected according to the jump vector $\vvec$, which defines a distribution over the nodes in the graph.
Typically, the jump probability is set to $\gamma = 0.15$, and the jump vector is set to the uniform vector.

The ``organic'' part of the random walk is governed by the transition matrix $\Pmatrix$, which  defines the transition probability $P[i,j]$ between any two nodes $i$ and $j$. The transition matrix is typically defined as the normalized adjacency matrix of graph $G$. Special care is required for the sink nodes in the graph, that is, nodes with no outgoing edges. In our work, we adopt the convention that, when at a sink node,  the random walk performs a jump to a node chosen uniformly at random~\cite{pagerank-survey}. That is, the corresponding zero-rows in the matrix $\Pmatrix$ are replaced by the uniform vector.

The Pagerank vector $\pvec$ satisfies the equation:
\begin{equation}
\pvec^T= (1-\gamma) \pvec^T \Pmatrix + \gamma \, \vvec^T
\label{eqn:PR}
\end{equation}
It can be computed either by solving the above equation, or by iteratively applying it to any initial probability vector.

The Pagerank algorithm is fully defined by the three parameters we described above: the transition matrix $\Pmatrix$,  the restart probability $\gamma$, and the restart (or jump) vector $\vvec$. Different settings for these parameters result in different algorithms.
Given a graph $G = (V,E)$, let $\allPR(G)$  denote the family of all possible Pagerank algorithms on graph $G$. Each algorithm in $\allPR(G)$ corresponds to a triplet $(\Pmatrix(G), \gamma, \vvec(G))$ for the parameters of the random walk. This is a very general family that essentially includes all possible random walks defined over the nodes of graph $G$.
We will refer to the algorithm that uses the typical settings as the \emph{original} Pagerank algorithm, $\OPR$, and use $\pvec_O$ to denote its pagerank vector.

Several variations of the original Pagerank algorithm have been proposed, that modify the above parameters to achieve different goals~\cite{pagerank-survey}.
We are interested in defining \emph{fair} Pagerank algorithms.

\subsection{Fair Pagerank}
We focus on graphs where a set of nodes defines a protected group based on the value of some protected attribute.
For example, in the case of social, collaboration, or citation networks where each node is a person, protected attributes may correspond to gender, age, race, or  religion. In the following for simplicity, we assume binary such attributes, but our algorithms can be extended for the general case.
We consider two types of nodes, red and blue nodes, and the corresponding groups denoted $R$ and $B$ respectively. Group $R$ is the protected group. 
We denote with $r = \frac{|R|}{n}$, and $b = \frac{|B|}{n}$, the fraction of nodes that belong to the red and  blue group respectively. 

Let $\PR\in\allPR(G)$ be a Pagerank algorithm on graph $G$. We will use $\PR(u)$ to denote the pagerank mass that $\PR$ assigns to node $u$, and, abusing the notation, $\PR(R)$ to denote the total pagerank mass that $\PR$ assigns to the nodes in the red group (for the blue group, $\PR(B) = 1-\PR(R))$. 
We will say that $\PR$  is \emph{fair}, if it assigns weights to each group according to a specified ratio $\phi$.

\begin{definition} [Pagerank Fairness]
	Given a graph $G = (V,E)$ containing the protected group $R\subseteq V$, and a value $\phi \in (0,1)$, a Pagerank algorithm $\mathrm{PR} \in \allPR(G)$ is $\phi$-fair on graph $G$, if $\mathrm{PR}(R) = \phi$. 
\end{definition}

The ratio $\phi$ may be specified so as to implement specific affirmative action policies,
or other fairness enhancing interventions.
For example, $\phi$ may be set in accordance to the 80-percent rule advocated by the US
Equal Employment Opportunity Commission (EEOC), or some other formulation of disparate impact \cite{disparate-impact}. 
Setting $\phi$ = $r$, we ask for a fair Pagerank algorithm that assigns weights proportionally to the sizes of the two groups.
In this case, fairness  is analogous to  demographic parity, i.e., the requirement that the demographics of those receiving a 
positive outcome are identical to the demographics of the
population as a whole \cite{fairness-awarness}.

Our goal is to define fair Pagerank algorithms. We say that a \emph{family} of Pagerank algorithms $\textsl{FPR} \subseteq \allPR(G)$ is $\phi$-fair if all the Pagerank algorithms in the family are $\phi$-fair. The first problem we consider is to find such families of algorithms.

\begin{problem}
Given a graph $G = (V,E)$ containing a protected group of nodes $R \subseteq V$, and a value $\phi \in (0,1)$, define a family of algorithms $\textsl{FPR} \subseteq \allPR(G)$ that is $\phi$-fair.
\end{problem}

We can achieve fairness by modifying the parameters of the Pagerank algorithm.
For the following, we assume the jump probability $\gamma$ to be fixed, and we only consider modifications to the transition matrix $\Pmatrix$ and the jump vector $\vvec$. A $\phi$-fair family of algorithms is defined by a specific process, parameterized by $\phi$, for defining $\Pmatrix$ and $\vvec$.

A fair Pagerank algorithm will clearly output a different weight vector than the original Pagerank algorithm.
We assume that the weights of the original Pagerank algorithm carry some \emph{utility}, and use these weights to measure the \emph{utility loss} for achieving fairness. Concretely, if $\fvec$ is the output of a fair Pagerank algorithm $\FPR$ and $\pvec_O$ is the output of the original Pagerank algorithm $\OPR$, we define the utility loss of $\FPR$ as: $L(\FPR) = L(\fvec,\pvec_O) = \|\fvec-\pvec_O\|^2$.

The second problem we consider is finding a fair algorithm that minimizes the utility loss.

\begin{problem}
	Given a $\phi$-fair family of algorithms $\textsl{FPR} \subset \allPR(G)$, find an algorithm $\mathrm{PR}$ $\in \textsl{FPR}$ that minimizes the utility loss $L(\mathrm{PR})$.
\end{problem}

Finally, we introduce an extension of the fairness definition, termed \emph{targeted fairness}, that asks for a fair distribution of
weights among a specific set of  nodes $S$ that is given as input. The subset $S$ contains a protected group of nodes $S_R$.
Targeted fairness asks that a fraction $\phi$ of the pagerank mass that $\PR$ assigns to $S$ goes to the protected group $S_R$. 
For example, assume a co-authorship graph $G$, where $S$ is the set of authors of a particular male-dominated field. We want to allocate to the female authors $S_R$ in this field a fraction $\phi$ of the total pagerank allocated to $S$.  

\begin{definition} [Targeted Fairness]
	Given a graph $G = (V,E)$, a subset of nodes $S\subseteq V$ containing a protected group $S_R\subseteq S$, and a value $\phi \in (0,1)$,  a Pagerank algorithm $\mathrm{PR}\in \allPR(G)$, is targeted $\phi$-fair on the subset $S$ of $G$, if $\mathrm{PR}(S_R) = \phi \mathrm{PR}(S)$. 
\end{definition}

The two problems we defined above can also be defined for targeted fairness.

%% file: sec3.tex
\section{Fairness Sensitive PageRank}
\label{sec:fairness-sensitive}
Our first family of algorithms achieves fairness by keeping the transition matrix $\Pmatrix$ fixed and changing the jump vector $\vvec$ so as to meet the fairness criterion. We denote this family of algorithms as the \emph{Fairness-Sensitive Pagerank} (\textsl{FSPR}) algorithms.

\subsection{The \textsl{FSPR} Algorithm}
First, we note that that pagerank vector $\pvec$ can be written as a linear function of the jump vector $\vvec$.
Solving Equation (\ref{eqn:PR}) for $\pvec$, 
we have that $\pvec^T = \vvec^T\Qmatrix$, where

\[
\Qmatrix = \gamma \left[\Imatrix - (1-\gamma)\Pmatrix \right]^{-1}
\]

Note that if we set $\vvec = \evec_j$, the vector with $\evec_j[j] = 1$ and zero elsewhere, then $\pvec^T = \Qmatrix_{j}^T$, the $j$-th row of matrix $\Qmatrix$. Therefore, the row vector $\Qmatrix_{j}^T$ corresponds to the personalized pagerank vector of node $j$. The pagerank vector $\pvec$ is a linear combination of the personalized pagerank vectors of all nodes, as defined by the jump vector.

Let $\pvec[R]$ denote the pagerank mass that is allocated to the nodes of the protected category. 
We have that 
$$
\pvec[R] =  \sum_{i\in R} \left(\vvec^T \Qmatrix \right) [i]  
		= \vvec^T \left( \sum_{i\in R} \Qmatrix_{i} \right) 
		= \vvec^T \Qmatrix_R 
$$
where $\Qmatrix_{i}$ is the $i$-th column of matrix $\Qmatrix$, and $\Qmatrix_R$ is the vector that is the sum of the columns of $\Qmatrix$ in the set $R$. $\Qmatrix_R[j]$ is the total personalized pagerank that node $j$ allocates to the red group. 

For the algorithm to be fair, we need $\pvec[R] = \phi$. 
Thus, our goal is to find a jump vector $\vvec$ such that $\vvec^T \Qmatrix_R = \phi$. Does such a vector always exist? We prove the following:

\begin{lemma}
\label{lemma:solution-feasibiliy}
Given the vector $\Qmatrix_R$, there exists a vector $\vvec$ such that $\vvec^T \Qmatrix_R  = \phi$, if and only if, there exist nodes $i,j$ such that $\Qmatrix_R[i] \leq \phi$ and $\Qmatrix_R[j] \geq \phi$
\end{lemma}

\begin{proof}
We have  
$
\pvec[R]= \vvec^T \Qmatrix_R = \sum_{j = 1}^N \vvec[j]\Qmatrix_R[j]
$,
with $0 \leq \vvec[j] \leq 1$. If $\vvec^T \Qmatrix_R  = \phi$, there must exist $i,j$ with  $\Qmatrix_R[i] \leq \phi$, and $\Qmatrix_R[j] \geq \phi$.
Conversely, if there exists two such entries $i,j$, then we can find values $\pi_i$ and $\pi_j$, such that 
$\pi_i \Qmatrix_R[i] + \pi_j \Qmatrix_R[j] = \phi$ and $\pi_i + \pi_j = 1$.
\end{proof}

\noindent {\bf Complexity.} Note that there is a very efficient way to compute the personalized pagerank, $\Qmatrix_R[j]$, that node $j$ allocates to the red group. We can add a red and a blue absorbing node in the random walk  and estimate the probability for each node $j$ to be absorbed by the corresponding absorbing node. This can be done in the time required for running Pagerank. Thus, it is possible to compute the $\Qmatrix_R$ vector without doing matrix inversion to compute $\Qmatrix$.

\subsection{Minimizing Utility Loss}
\label{sec:optimization-fair-sensitive}
An implication of Lemma~\ref{lemma:solution-feasibiliy} is that, in most cases, there are multiple jump vectors that give a fair pagerank vector. We are interested in the solution that minimizes the utility loss. 

To solve this problem we exploit the fact that the utility loss function $L(\pvec_\vvec,\pvec_O) = \|\pvec_\vvec - \pvec_O\|^2$, where $\pvec_\vvec$ is the fair pagerank vector and $\pvec_O$ the original vector, is convex. We also can express the fairness requirement as a linear constraint. Thus, we define the following convex optimization problem.

\begin{equation*}
	\begin{aligned}
		& \underset{\xvec}{\text{minimize}} & & \|\xvec^T \Qmatrix - \pvec_O\|^2  \\
		& \text{subject to} & & \xvec^T \Qmatrix_R  = \phi\\
		& & & \sum_{i = 1}^{n} \xvec[i] = 1 \\
		& & & 0 \leq \xvec[i] \leq 1, \; i = 1, \ldots, n \\
	\end{aligned}
\end{equation*}

This problem can be solved using standard convex optimization solvers. 
In our work, we used the CVXOPT software package\footnote{https://cvxopt.org/}. 
The complexity of the algorithm is dominated by the computation of matrix $\Qmatrix$ which requires a matrix inversion. We can speed up this process by exploiting the fact that the rows of $\Qmatrix$ are personalized pagerank vectors, which can be computed (in parallel) by performing multiple random walks. We can improve performance further using approximate computation, e.g., \cite{approximatepr}.

\subsection{Targeted Fairness \textsl{FSPR} Algorithm}
We can formulate a similar convex optimization problem for the targeted fairness problem.
Let $\Qmatrix_S = \sum_{i\in S} \Qmatrix_i$ be the sum of columns of $\Qmatrix$ for the nodes in $S$, and $\Qmatrix_{S_R} = \sum_{i\in S_R} \Qmatrix_i$be  the sum of columns of $\Qmatrix$ for the nodes in $S_R$. We define a convex optimization problem that is exactly the same as in Section~\ref{sec:optimization-fair-sensitive}, except for the fact that we replace the constraint $\xvec^T \Qmatrix_R  = \phi$ with the constraint $\xvec^T \Qmatrix_{S_R}  = \phi \xvec^T \Qmatrix_S $.

%% file: sec4.tex
\section{Locally Fair PageRank}
\label{sec:local-fair}

Our second family of fair Pagerank algorithms,  termed \emph{Locally Fair Pagerank}
(\textsl{LFPR}), takes a microscopic view of fairness,
by asking that \textit{each individual node} acts fairly, i.e., each node distributes its own pagerank to red and blue nodes fairly.
In random walk terms, local fairness defines a dynamic process 
that can be viewed as a random walk that is
fair at each step, and not just at convergence.

The \textsl{LFPR} contains all Pagerank algorithms, where all rows of the transition matrix $\Pmatrix$ are $\phi$-fair vectors. That is, for every node $i \in V$, $\sum_{j \in R} P[i,j] = \phi$. Also, the jump vector $\vvec$ is $\phi$-fair: $\sum_{j \in R} \vvec[j] = \phi$.

We now consider specific algorithms from the family of locally fair algorithms.

\subsection{The Neighborhood \textsl{LFPR} Algorithm}
We first consider a node that treats its neighbors fairly by allocating a fraction $\phi$ of its pagerank to its red neighbors
and the remaining $1 - \phi$ fraction to its blue neighbors.
In random walk terms, at each node the probability of transitioning to a red neighbor is $\phi$ and the probability of transitioning to a blue neighbor $1-\phi$.

Formally, we define the \textit{neighborhood locally fair pagerank} ({\neighborlocal})  as follows.
Let $out_R(i)$ and $out_B(i)$  be the number of  outgoing edges  from node $i$ to red nodes and blue nodes respectively.  We define $\Pmatrix_R$ as the stochastic transition matrix that handles transitions to red nodes, or random jumps to red nodes
if such links do not exist:
\[
\Pmatrix_R[i,j] = \left \{
\begin{tabular}{cl}
$\frac{1}{out_R(i)}$, & if $(i, j)$ $\in$ E and $j$ $\in$ $R$   \\
$\frac{1}{|R|}$, & if $out_R(i) = 0$ and $j$ $\in$ $R$  \\
0, & otherwise \\
\end{tabular}
\right.
\]
The transition matrix 
$\Pmatrix_B$ for the blue nodes is defined similarly. The transition matrix $\Pmatrix_N$ of the {\neighborlocal} algorithm is defined as:

\begin{center}
$\Pmatrix_N = \phi \, \Pmatrix_R + (1 - \phi) \, \Pmatrix_B
$
\end{center}

We also define a $\phi$-fair jump vector  $\vvec_N$ with $\vvec_N[i]$ = $\frac{\phi}{|R|}$, if $i \in R$, and
$\vvec_N[i]$ = $\frac{1-\phi}{|B|}$, if $i \in B$.
The neighborhood locally-fair pagerank vector $\pvec_N$ is defined as:

\[
\pvec_N^T = (1 - \gamma) \pvec_N^T \Pmatrix_N + \gamma \, \vvec_N^T
\]

\subsection{The Residual-based \textsl{LFPR} Algorithms}
We consider an alternative fair behavior for individual nodes.  
Similarly to the {\neighborlocal} algorithm,   each node $i$ acts fairly by respecting the $\phi$ ratio when distributing its own pagerank to red and blue nodes. However, now node $i$
treats its neighbors the same, independently of their color and assigns to each of them the same portion of its pagerank. When a node is in a ``biased'' neighborhood, i.e., the ratio of its red neighbors is different than
$\phi$, to be fair, node $i$ distributes only a fraction of its pagerank to its neigbors, and the remaining portion of its pagerank
to nodes in the underrepresented group. We call the remaining portion \emph{residual} and denote it by $\delta(i)$.
How $\delta(i)$ is distributed to the underrepresented group is determined by a \textit{residual policy}.

Intuitively, this corresponds to a fair random walker that upon arriving at a node $i$,  with  probability 1-$\delta(i)$ follows one of $i$'s outlinks and with probability $\delta(i)$
jumps to one or more node belonging to the group that is locally underrepresented.

We now describe the algorithm formally. 
We divide the (non sink) nodes in $V$ into two sets, $L_R$ and $L_B$, based on the fraction of their red and blue neighbors.
Set $L_R$ includes 
all nodes $i$ such that
$out_R(i)/out(i) < \phi$, where $out(i)$ the out-degree of node $i$,
that is, the nodes for which the ratio of red nodes in their neighborhoods is smaller than the required $\phi$ ratio. These nodes have a residual that needs to be distributed to red nodes.
Analogously, $L_B$ includes 
all nodes $i$ such that
$out_R(i)/out(i) \geq \phi$,
that have a residual to be distributed to blue nodes.

Consider a node $i$ in $L_R$. 
To compute $\delta_R(i)$ note that each neighbor of $i$ gets a fraction $\rho_R(i) = \frac{1-\delta_R(i)}{out(i)}$ of $i$'s pagerank. The residual  $\delta_R(i)$ of $i$'s pagerank  goes to red nodes. In order for node $i$ to be $\phi$-fair, we have:
\begin{equation}
\frac{1-\delta_R(i)}{out(i)}out_R(i) + \delta_R(i)= \phi \label{eq:excess1}
\end{equation}
Solving for $\delta_R(i)$ and using the fact that $out(i) = out_R(i)+out_B(i)$ we get $\delta_R(i) = \phi - \frac{(1-\phi)\,out_R(i)}{out_B(i)}$, and $\rho_R(i) = \frac{1-\phi}{out_B(i)}$.

Analogously, for a node $i$ in $L_B$, we get a residual $\delta_B(i) = (1 -\phi) - \frac{\phi\,out_B(i)}{out_R(i)}$ that goes to the blue nodes, and $\rho_B(i)$ = $\frac{\phi}{out_R(i)}$.

For a sink node $i$, we assume that $i$ belongs to both $L_R$ and $L_B$ with residual $\delta_R(i)$ = $\phi$ and $\delta_B(i)$ = $1- \phi$.

\vspace*{0.1in}
\noindent \textit{Example.}
Consider a node $i$ with 5 out-neighbors, 1 red and 4 blue, and let $\phi$  be $0.5$.
This is a ``blue-biased''node, that is, $i$ $\in$ $L_R$.
With the residual algorithm, each of $i$'s  neighbors gets
$\rho_R(i)$ = $0.5/4 = 1/8$ portion of $i$'s pagerank, resulting in red neighbors getting $1/8$ and blue neighbors  $4/8$ of i's pagerank. The residual $\delta_B(i)$ = $3/8$ goes to nodes in the red group so as to attain the $\phi$ ratio and make $i$ fair. 
In terms of the random walker interpretation, a random walker that arrives at $i$, with probability $5/8$ chooses one of $i$'s outlinks uniformly at random and
with probability $3/8$ jumps to nodes in the red group. \hfill$\qed$.


Formally, we define $\Pmatrix_L$ 
as follows:

\[
\Pmatrix_L[i, j] = \left \{
\begin{tabular}{cl}
$\frac{1-\phi}{out_B(i)}$, & if $(i, j)$ $\in$ $E$ and  $i \in  L_R$\\
$\frac{\phi}{out_R(i)}$, & if $(i, j)$ $\in$ $E$ and  $i \in L_B$ \\
$0$ & otherwise
\end{tabular}
\right.
\]

%
$\Pmatrix_L$ handles the transitions of nodes to their neighborhood.
To express the residual distribution policy,  we introduce matrices $\Xmatrix$ and $\Ymatrix$, that 
capture the policy for distributing the residual to red and blue nodes respectively.
Specifically,  $\Xmatrix[i, j]$  denotes the portion of the $\delta_R(i)$ of node $i$ $\in$ $L_R$ that goes to node $j$ $\in$ $R$, and
$\Ymatrix[i, j]$  the portion of the $\delta_B(i)$ of node $i$ $\in$ $L_B$ that goes to node $j$ $\in$ $B$.

The  locally-fair pagerank vector $\pvec_L$ is defined as:
\[
\pvec_L^T = (1 - \gamma) \pvec_L^T \, (\Pmatrix_L +  \Xmatrix +  \Ymatrix) + \gamma \, \vvec_N^T
\]


\noindent \textbf{Residual Distribution Policies.} The $\Xmatrix$ and $\Ymatrix$ allocation matrices allow us the flexibility to 
specify appropriate policies for distributing the residual. 
For example, 
%
the {\neighborlocal} algorithm is a special case of the residual-based algorithm, with
\vspace*{-0.05in}
\[
\Xmatrix_N[i, j] = \left \{
\begin{tabular}{cl}
$\frac{\delta_R(i)}{out_R(i)}$, &  if $i \in L_R$, $(i,j) \in E$, and $j \in R$ \\
$\frac{\delta_R(i)}{|R|}$, &  if $i \in L_R$, $out_R(i) = 0$, and $j \in R$ \\
0 &  otherwise \\
\end{tabular}
\right.
\]

and $\Ymatrix_N[i, j]$ defined analogously.



We also consider residual policies where all nodes follow the same policy in distributing their residual, that is, each red node gets the same portion of the red residuals
and each blue node the same portion and the blue residuals. In this case, the residual policy is expressed through two (column) vectors $\xvec$ and $\yvec$, with
$\xvec[j]$ = 0 if $j \in B$ and $\yvec[j]$ = 0, $j \in R$.
Each node $i \in L_R$ sends a fraction $\delta_R(i)\xvec[j]$ of its pagerank to red node $j$, while each node $i \in L_B$ sends  a fraction $\delta_B(i)\yvec[j]$ of its pagerank to blue node $j$.

%
Let $\dvec_R$ be the vector carrying the red residual, and 
$\dvec_B$ the vector carrying the blue residual.
We have: 
\begin{equation*}
\pvec_L^T = (1 - \gamma)\pvec_L^T\, (\Pmatrix_L +  \dvec_R  \, \xvec^T + \dvec_B \, \yvec^T) + \gamma \, \vvec_N^T.
\end{equation*}

We define two locally fair {\pagerank} algorithms based on two intuitive policies of distributing the residual:

\vspace{0.05in}
\noindent
The \textit{Uniform Locally Fair {\pagerank}} ({\uniformlocal}) algorithm distributes the residual uniformly. Specifically, 
we define the vector $\xvec$,  as $\xvec[i]$  = $\frac{1}{|R|}$  for $i \in R$,
and the vector $\yvec$, as $\yvec[i]$ =	$\frac{1}{|B|}$, for $i \in B$.

\vspace{0.05in}
\noindent
The \textit{Proportional Locally Fair {\pagerank}} ({\proportionallocal}) algorithm distributes the residual proportionally to the original pagerank weights $\pvec_O$
Specifically, 
we define the vector $\xvec$, as
$\xvec[i]$ = 	$\frac{ \pvec_O[i]  }{\sum_{i \in R}\pvec_O[i]}$, for $i \in R$, 
and the vector  $\yvec$, as $\yvec[i]$ = 
$\frac{ \pvec_O[i] }{\sum_{i \in B}\pvec_O[i]}$, for $i \in B$.





\subsection{Fairness of the \textsl{LFPR} Algorithms}
Although each node acts independently of the other nodes in the network,
this  microscopic view of  fairness 
results in a macroscopic view of fairness. Specifically, we prove the following theorem.

\begin{theorem}
	The  locally fair {\pagerank} algorithms are $\phi$-fair. \label{theorem:local}
\end{theorem} 


%

\begin{proof}
Let $\evec_R$ denote the vector with 1's at the red nodes and zero at the blue nodes. The amount of pagerank that vector $\pvec_L$ gives to the red nodes can be expressed as $\pvec_L^T\evec_R$. Let $\Pmatrix_D$ = $\Pmatrix_L + \Xmatrix + \Ymatrix$, we have:
\[
\pvec_L^T\evec_R= (1-\gamma)\pvec_L^T \Pmatrix_D\evec_R + \gamma \vvec_N^T\evec_R
\]
By design we have that $\vvec_N^T\evec_R = \phi$. For transition matrix $\Pmatrix_D$, due to the local fairness property, we know that each row has probability $\phi$ of transitioning to a red node. Therefore, $\Pmatrix_D\evec_R = \phi\evec$, where $\evec$ is the vector of all ones. Note that since $\pvec_L$ is a distribution we have that $\pvec_L^T\evec = 1$. We thus have:
$\pvec_L^T\evec_R= (1-\gamma)\phi + \gamma\phi = \phi$.
\end{proof}





\subsection{Minimizing Utility Loss}
We now consider how to distribute the residual so as to minimize the utility loss of the locally fair Pagerank. We denote this algorithm as {\residuallocal}. To this end,
we compute the $\xvec$ and $\yvec$  residual distribution vectors by formulating an optimization problem.

We can write the vector $\pvec_L$ as a function of the vectors $\xvec$ and $\yvec$ as follows:
\[
\pvec_L^T(\xvec,\yvec) = \gamma \, \vvec_N^T \left[\Imatrix - (1-\gamma)(\Pmatrix_L + \delta_R\, \xvec^T + \delta_B\, \yvec^T) \right]^{-1}
\]
We can now define the optimization problem of finding the vectors $\xvec$ and $\yvec$ that minimize the loss function $L(\pvec_L, \pvec_O) = \|\pvec_L(\xvec,\yvec) - \pvec_O\|^2$ subject to the constraint that the vectors $\xvec$ and $\yvec$ define a distribution over the nodes in $R$ and $B$ respectively.
%
%

Since our problem is not convex, we implement a Stochastic Random Search algorithm for solving it, that looks at randomly selected directions for the gradient, and selects the one that causes the largest decrease.
We enforce the distribution constraints
by adding a penalty term 
$\lambda$ $\left ( (\sum_{i = 1}^{n} \xvec_i - 1)^2 + (\sum_{i = 1}^{n} \yvec_i - 1)^2\right )$. 
We enforce the positivity constraints
through proper bracketing at the line-search step.
The complexity of the algorithm is $O\left ( I\cdot K\cdot T_{PR}\right )$, where $I$ is the number of iterations, $K$ the number of random directions that are examined, and $T_{PR}$ the cost of running Pagerank. In our experiments, $I$ and $K$ are in the order of a few hundreds. 



\subsection{Targeted Fairness \textsl{LFPR} Algorithms}      
We can apply the locally fair algorithms 
to the targeted fairness problem. 
Let $S_R$ and $S_B$ be the red and blue nodes in the set $S$ respectively, and let $I_S$ be the set of in-neighbors of $S$.
The idea is that the nodes in $I_S$ should distribute their pagerank to $S_R$ and $S_B$ fairly, 
such that the 
portion that goes to nodes in  $S_R$ is a $\phi$ fraction of the total pagerank that goes to the set $S$.
We can implement the same redistribution policies as in the case of the neighborhood local and the residual-based local fair algorithms. 

We also need the (global) jump vector $\vvec$ to obey the $\phi$ ratio for the nodes in $S$. We can achieve this by redistributing the probability $|S|/n$ of the jump vector according to the $\phi$ ratio.
%
%



%% file: sec5.tex
\section{Personalized Fairness}
\label{sec:universal}

A special case of the {\pagerank} algorithm is when the restart vector is defined to be the unit vector $\evec_i$ that puts all the mass at a single node $i$. For any {\pagerank} algorithm $\PR \in \allPR$, we can define the corresponding personalized algorithm $\PR_i$ by setting $\vvec = \evec_i$. 

The output of the algorithm  $\PR_i$ is a probability vector, where $\PR_i(u)$ is the probability that a random walk that always restarts at node $i$ is at $u$ after infinite number of steps. We say that node $i$ allocates this probability to node $u$. Personalized random walks have found several applications in network analysis~\cite{pagerank-survey}. For example, the probability $\PR_i(u)$ can be seen as a measure of proximity between node $i$ and node $u$, and it has been used for recommendations. 

For a personalized Pagerank algorithm $\PR_i$, we define $\PR_i(R)$ and $\PR_i(B)$ to be the probability that node $i$ allocates to the red and blue groups respectively. Recall that if $\vvec$ is the jump vector, then $\PR(R) = \sum_{i\in V} \vvec[i] \PR_i(R)$. We can think of the probability $\PR_i(R)$, as a measure of how a specific node $i$ ``views'' the red group, while $\PR(R)$ captures the value that the network places on the red nodes on aggregate.

We could define fairness for the $\PR_i$ algorithm using the standard fairness definition. However, note that since the random walk jumps with probability $\gamma$ to node $i$ at every step, this immediately adds probability $\gamma$ to the category of node $i$. This probability is due to the random jump and not due to the ''organic'' random walk, and the structure of the graph.  We thus subtract this probability, and we define the vector $\overline{\PR_i}$, where $\overline{\PR_i}(i) = \PR_i(i)-\gamma$, and $\overline{\PR_i}(u) = \PR_i(u)$, for $u \neq i$.  Another way to think of this is that an amount $\gamma$  of probability is reserved for restarting, and the remaining $1-\gamma$ is allocated through the random walk. This is the probability mass that we want to allocate fairly. We say that the \emph{personalized} Pagerank algorithm $\PR_i$ is $\phi$-fair if $\overline{\PR_i}(R) = \phi(1-\gamma)$.

Intuitively, fairness of $\PR_i$ implies that node $i$ treats the red and blue groups fairly.  For example, if we think of $\PPR(R)$ as a proximity measure, and $\phi = 0.5$, $\phi$-fairness implies that node $i$ is equally close to the red and blue groups. Given that this probability is often used for recommendations, this has also implications to the fairness of the recommendations.



\begin{figure}[]
	\includegraphics[width = \columnwidth]{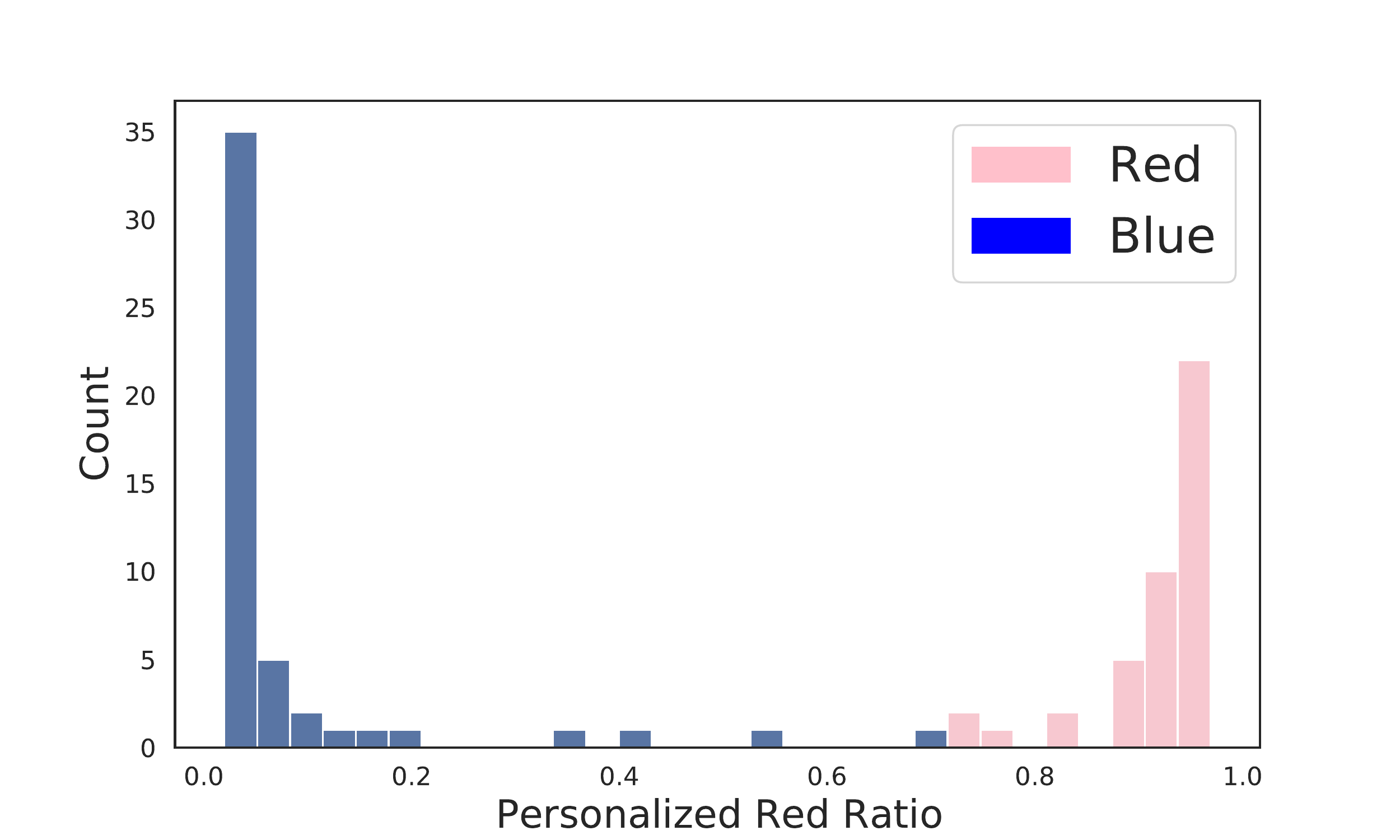}
	\caption{Histogram of personalized red ratio values.}
	\label{books-histogram}
\end{figure}


Note that a Pagerank algorithm $\PR$ may be fair, while the corresponding personalized Pagerank algorithms are not. In Figure~\ref{books-histogram}, we consider the original Pagerank algorithm $\OPR$, and we show the histogram of the $\overline{\OPR_i}(R)$ values for the \textsc{books} dataset (described in Section~\ref{sec:experiments}), for the red and blue nodes, in red and blue respectively. For this dataset, we have that $r = 0.47$ and $\OPR(R)$ is 0.46. That is, the original Pagerank algorithm is almost fair for $\phi = r$. However, we observe that the distribution of the $\overline{\OPR_i}(R)$ values is highly polarized. 
The values for the red nodes (shown in red) are concentrated in the interval $[0.8,1]$, while the values for the blue nodes (in blue) are concentrated in the interval $[0,0.2]$. Thus, although the network as a whole has a fair view of the two groups, the individual nodes are highly biased in favor of their own group.

Motivated by this observation, we consider a stronger definition of fairness, where given an algorithm $\PR$ we require that \emph{all} derivative Personalized Pagerank versions of this algorithm are fair. That is, it is not sufficient that the algorithm treats the red group fairly on aggregate, but we require that each individual node is also fair. 

\begin{definition}[Universal Personalized Fairness]
	Given a graph $G = (V,E)$ containing the protected group $R\subseteq V$, and a value $\phi \in (0,1)$, a Pagerank algorithm $\PR \in \allPR(G)$ is universally personalized $\phi$-fair on graph $G$, if for every node $i \in V$, the personalized Pagerank algorithm $\PR_i$ is personalized $\phi$-fair.
\end{definition}

Since we want all personalized algorithms to be fair, universal personalized fairness (universal fairness for short) is a property of the transition matrix $\Pmatrix$ of the Pagerank algorithm. Since the \textsl{FSPR} family does not change the matrix $\Pmatrix$, it is not universally fair. We will show that the locally fair algorithms are universally fair. Furthermore, we can prove that universally fair algorithms are locally fair. Therefore, universal fairness is equivalent to local fairness.
 

\begin{theorem}
	A Pagerank algorithm $\PR$ is universally personalized $\phi$-fair if and only if it is locally fair.
\end{theorem}

\begin{proof}
	We first prove that if an algorithm $\PR$ is locally fair then it is  personalized fair. The proof is similar to that of Theorem~\ref{theorem:local}.
	Let $\pvec_i$ denote the personalized pagerank vector of algorithm $\PR_i$. We know that 
	$
	\pvec_i^T = (1-\gamma)\pvec_i^T\Pmatrix + \gamma \evec_i^T
	$
	where $\evec_i$ is the vector with 1 at the position $i$, and 0 everywhere else.
	The amount of probability that $\PR_i$ allocates to the red category can be computed as $\PR_i(R) = \pvec_i^T \evec_R$, where $\evec_R$ is the vector with 1 at the positions of all red nodes and 0 everywhere else. Multiplying the equation for $\pvec_i^T$ with $\evec_R$ we have:
	\[
	\PR_i(R) = (1-\gamma)\pvec_i^T\Pmatrix\evec_R + \gamma \evec_i^T\evec_R
	\]
	Since $\PR$ is fair, for every row of  the transition matrix $\Pmatrix$, the probability of transitioning to a red node is $\phi$. Therefore we have that $\Pmatrix \evec_R = \phi\evec$, where $\evec$ is the vector of all 1's. Also, since $\pvec_i$ defines a distribution $\pvec_i^T\evec = 1$. Therefore $(1-\gamma)\pvec_i^T\Pmatrix\evec_R = \phi(1-\gamma)$. We have: 
	\[
	\PR_i(R) = \phi(1-\gamma) + \gamma\evec_i^T \evec_R 
	\]
	The value of the second term $\gamma\evec_i^T \evec_R$ depends on whether the node $i$ is red or blue. If $i$ is blue, $\gamma\evec_i^T \evec_R = 0$, and we have $\PR_i(R) = \phi(1-\gamma)$. If $i$ is red, $\gamma\evec_i^T \evec_R = \gamma$, and thus $\PR_i(R) = \phi(1-\gamma) + \gamma$, which proves our claim.
	
	For the opposite direction we make use of the fact that the pagerank vector can be written as $\pvec^T = \vvec^T\Qmatrix$, where $\vvec$ is the jump vector and $\Qmatrix = \gamma\left[\Imatrix - (1-\gamma)\Pmatrix\right]^{-1}$.
	%
	%
	From Section~\ref{sec:fairness-sensitive} we know that 
	the $i$-th row of matrix $\Qmatrix$ is equal to the personalized pagerank vector $\pvec_i$. The product $\rvec = \Qmatrix\evec_R$ is a vector where $\rvec[i] = \PR_i(R)$. We have assumed that the $\PR$ algorithm is universally personalized $\phi$-fair. Therefore $\rvec[i] = \phi(1-\gamma) + \gamma$ if $i$ is red, and $\rvec[i] = \phi(1-\gamma)$ if $i$ is blue. That is, $\rvec = \phi(1-\gamma)\evec + \gamma\evec_R$. Using the fact that $\rvec = \Qmatrix\evec_R$, and that $\Qmatrix\evec = \evec$, since $\Qmatrix$ is stochastic, we have the following derivations:
	\begin{align}
		\Qmatrix\evec_R & = \phi(1-\gamma)\Qmatrix\evec + \gamma\evec_R \nonumber\\
		\Qmatrix^{-1}\Qmatrix\evec_R & = \phi(1-\gamma)\Qmatrix^{-1}\Qmatrix\evec + \gamma\Qmatrix^{-1}\evec_R \nonumber\\
		\evec_R & = \phi(1-\gamma)\evec 
		 	+ \gamma\frac1\gamma \left(\Ivec - (1-\gamma)\Pmatrix \right)\evec_R \nonumber\\
		\evec_R & = \phi(1-\gamma)\evec + \evec_R - (1-\gamma) \Pmatrix \evec_R \nonumber\\
		\Pmatrix \evec_R & = \phi \evec \nonumber
	\end{align}
	The last equation means that the probability of transitioning from any node in the graph to a red node is $\phi$, which proves our claim. 
\end{proof}

The theorem holds also when we consider targeted fairness. We can prove that an algorithm is universally personalized targeted fair, if and only if it is locally fair. We omit the details of the proof due to lack of space.

%% file: sec6.tex
\begin{table*}[ht]
	\centering
	\caption{Real dataset characteristics. $r$, $b$: relative size, $\textit{cross}_R$, $\textit{cross}_B$: percentage of cross-edges, $p_R$, $p_B$: original pagerank assigned to the protected and unprotected group respectively.}
		\begin{tabular}{c c c c c c c c c c }
		\toprule
		{\textbf{Dataset}} & \#nodes & \#edges &Protected attribute & $r$ & $b$ & $\textit{cross}_R$ & $\textit{cross}_B$ & $p_R$ & $p_B$ \\
		\midrule
		\sc{books} &  92  & 748 & political (left) & 0.47 & 0.53 & 0.063 & 0.065& 0.46 & 0.54 \\
		\sc{blogs} & 1,222  & 19,089 & political (left) & 0.48 & 0.52 & 0.284 & 0.036&  0.33 & 0.67 \\
		\sc{dblp} & 13,015 & 79,972 & gender (women) & 0.17 & 0.83 & 0.96 & 0.86& 0.16  & 0.84\\
		\sc{twitter} & 18,470 & 61,157 & political (left) & 0.61 & 0.39 & 0.07 & 0.03& 0.57 & 0.43 \\
		\bottomrule
	\end{tabular}
	\label{table:real}
\end{table*}

\begin{figure}[]
	\centering
	\subfigure[{symmetric}]{
		{\includegraphics[width = 0.22\textwidth]{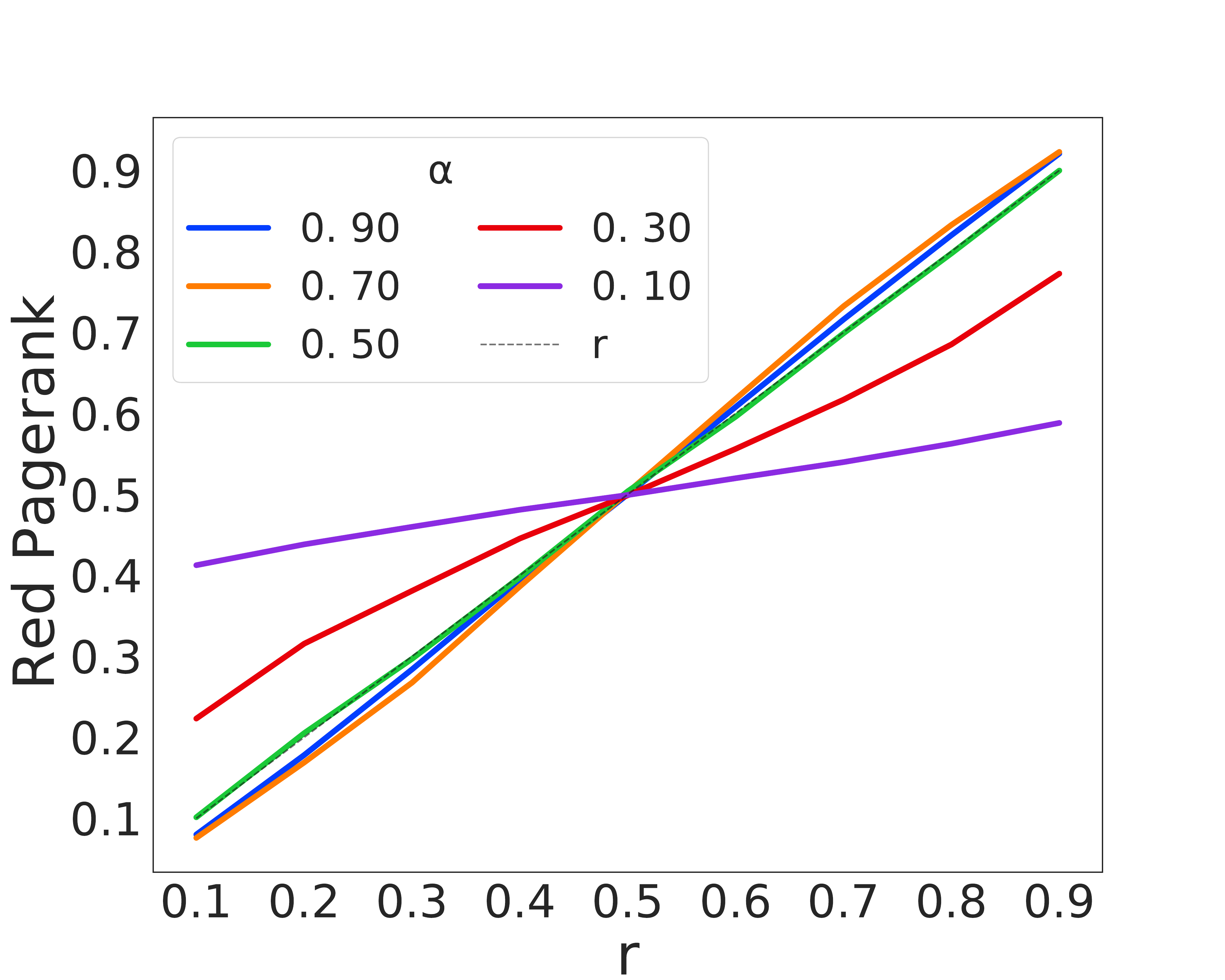}}
	}
	\subfigure[asymmetric]{
	{\includegraphics[width = 0.22\textwidth]{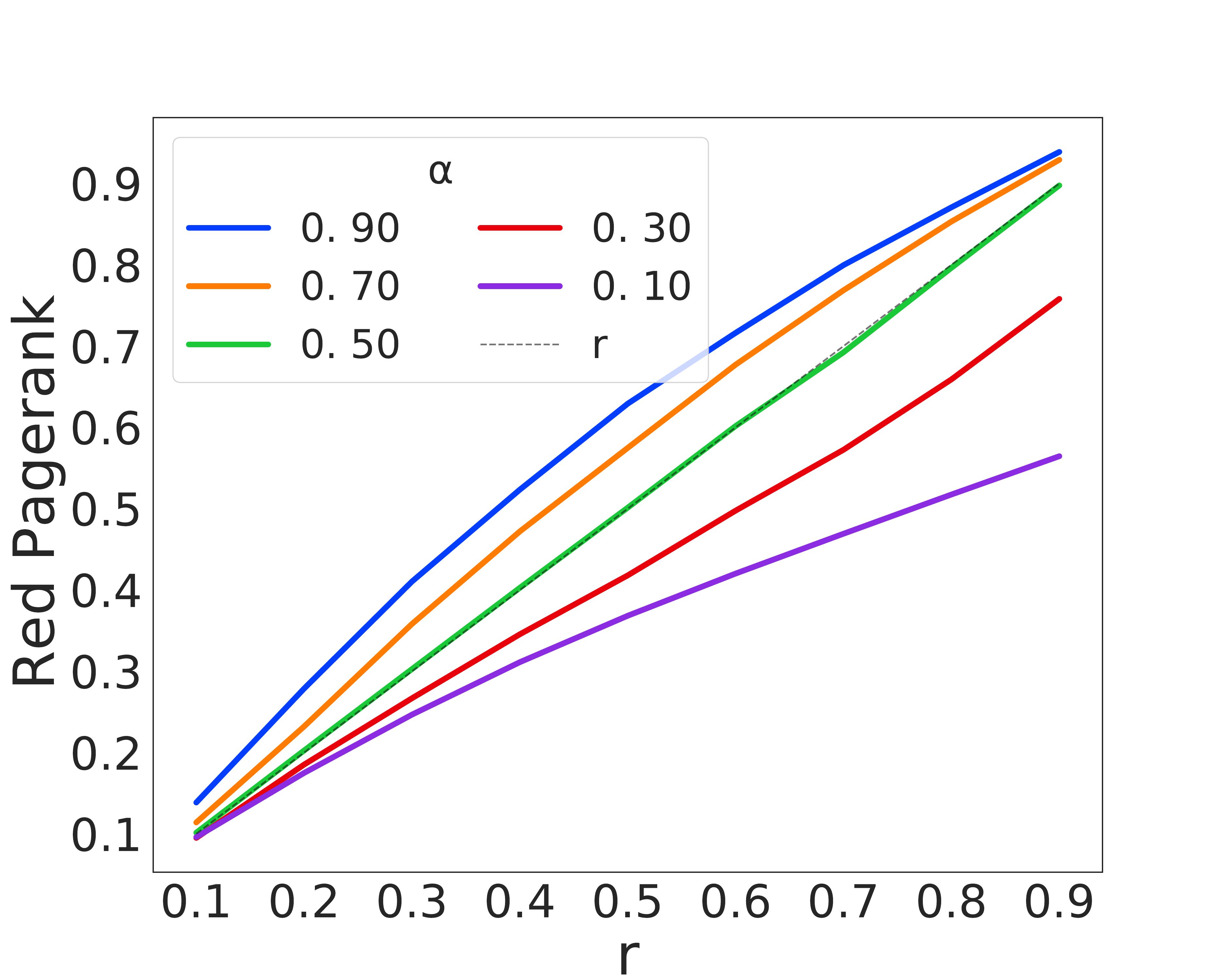}}
	}
	\caption{Red  {\pagerank}  with varying $r$ and $\alpha$, $\phi$-fairness corresponds to the identical line.}
	\label{fig:local-synthetic-all}
\end{figure}

\begin{figure*}[]
	\centering
	\subfigure[{$r$ = 0.1}]{
      {\includegraphics[width = 0.32\textwidth]{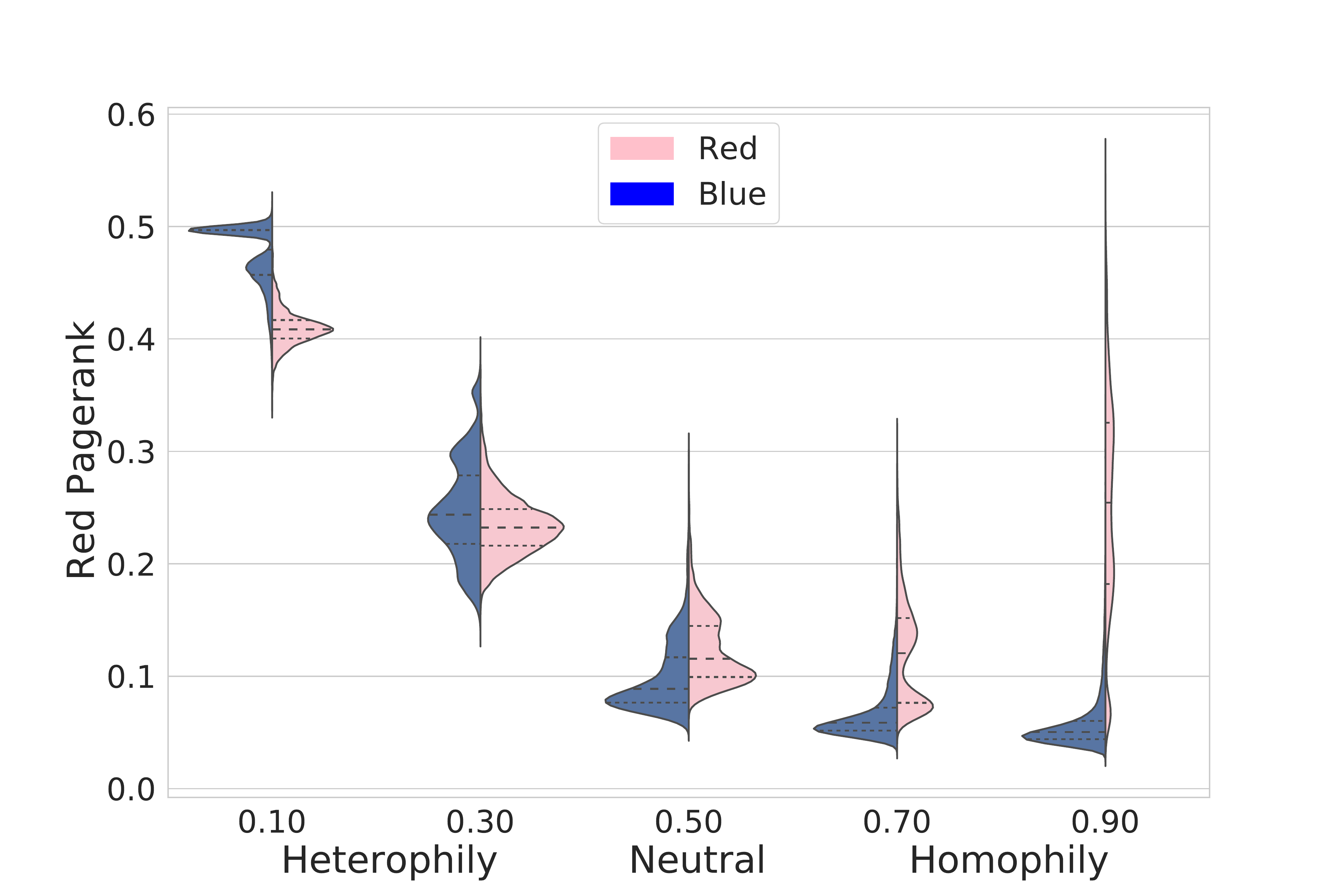}}
	}
	\subfigure[{$r$ = 0.3}]{
		{\includegraphics[width = 0.32\textwidth]{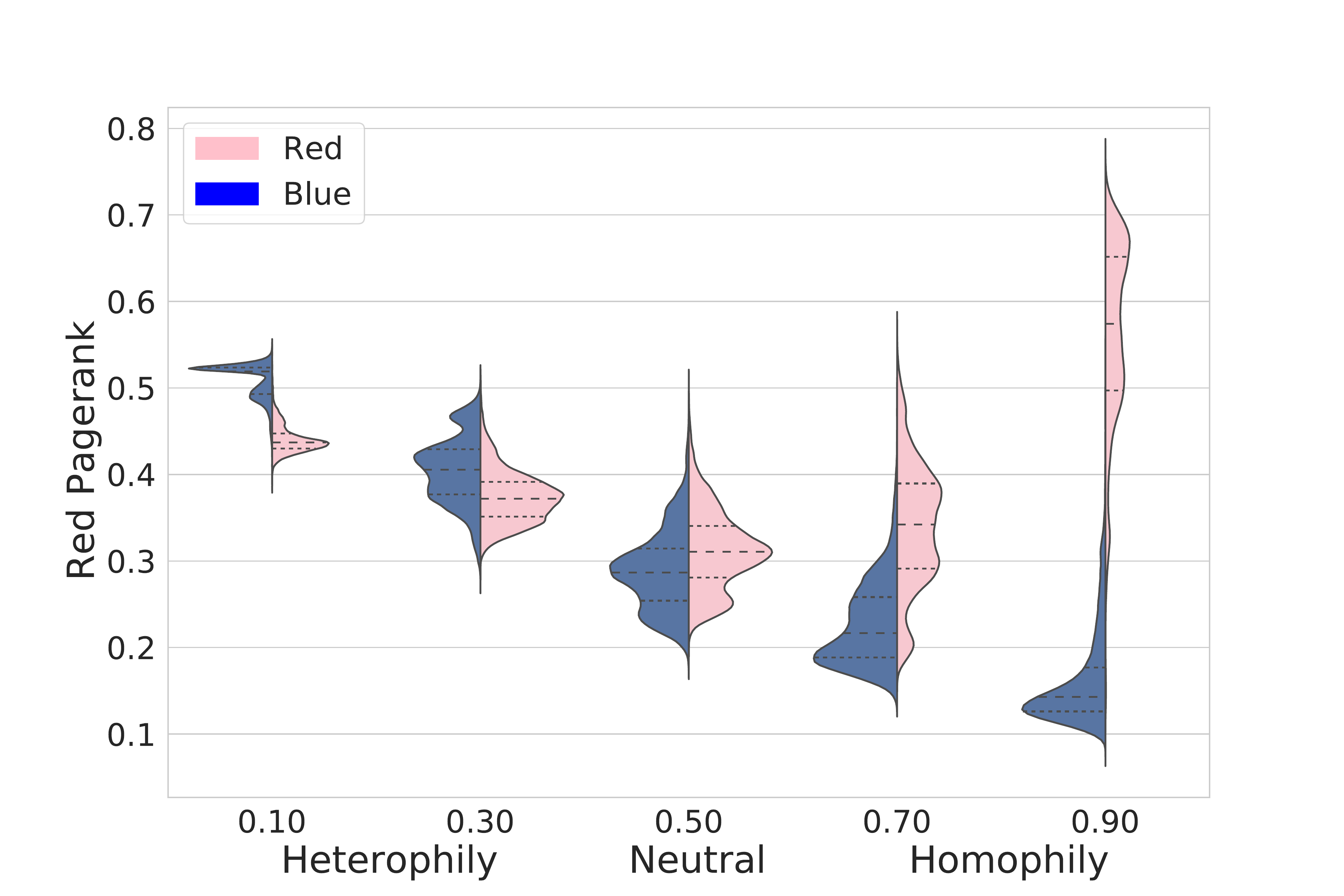}}
	}
	\subfigure[{$r$ = 0.5}]{
	{\includegraphics[width = 0.32\textwidth]{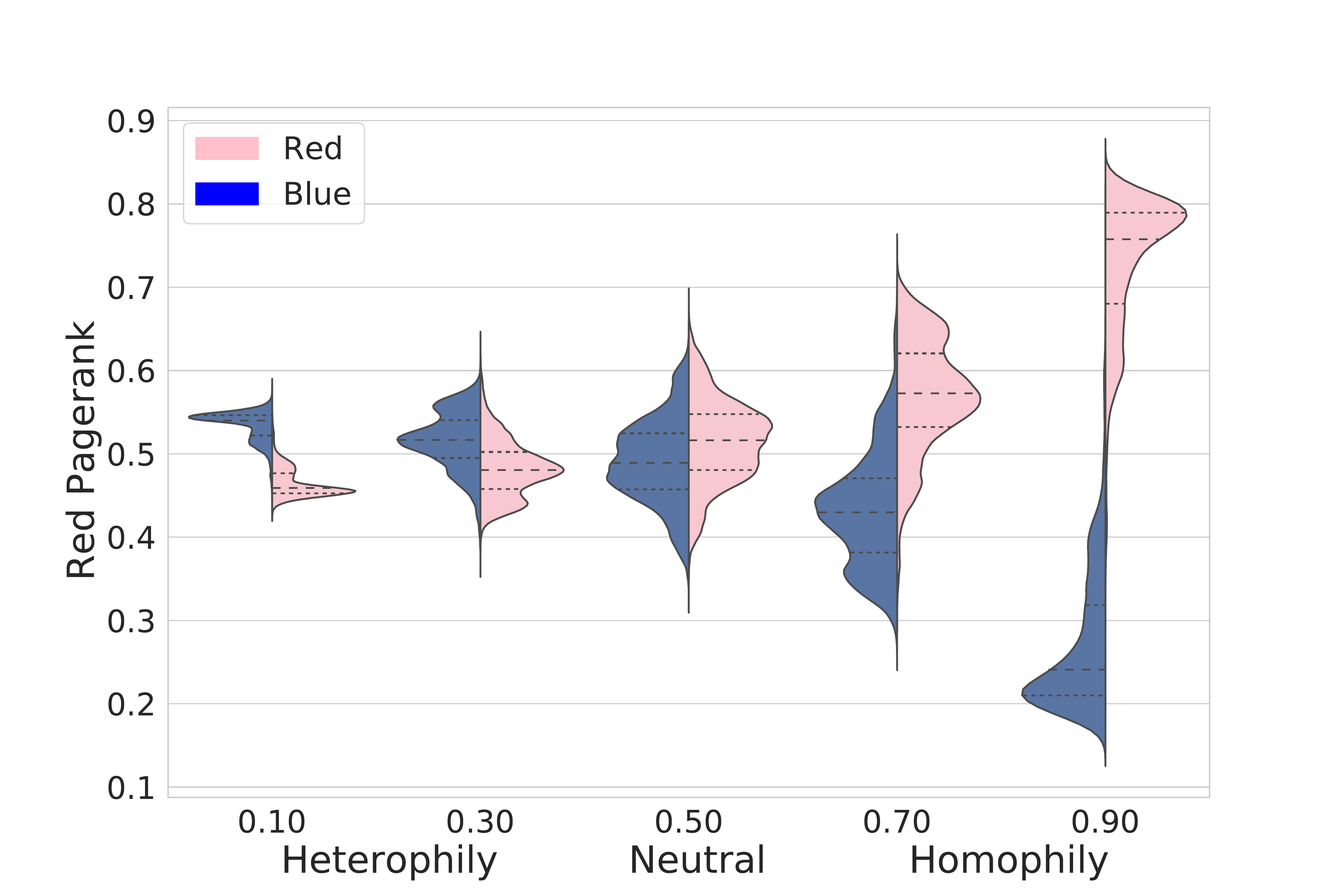}}
	}
	\caption{Distribution of the red personalized pagerank of the red and blue nodes with varying $\alpha$ in the symmetric case, $\phi$-fairness achieved when the red pagerank is $r$.}
	\label{fig:sym-violin-plots}
\end{figure*}

\begin{figure*}[]
	\centering
	\subfigure[{$r$ = 0.1}]{
  	{\includegraphics[width = 0.32\textwidth]{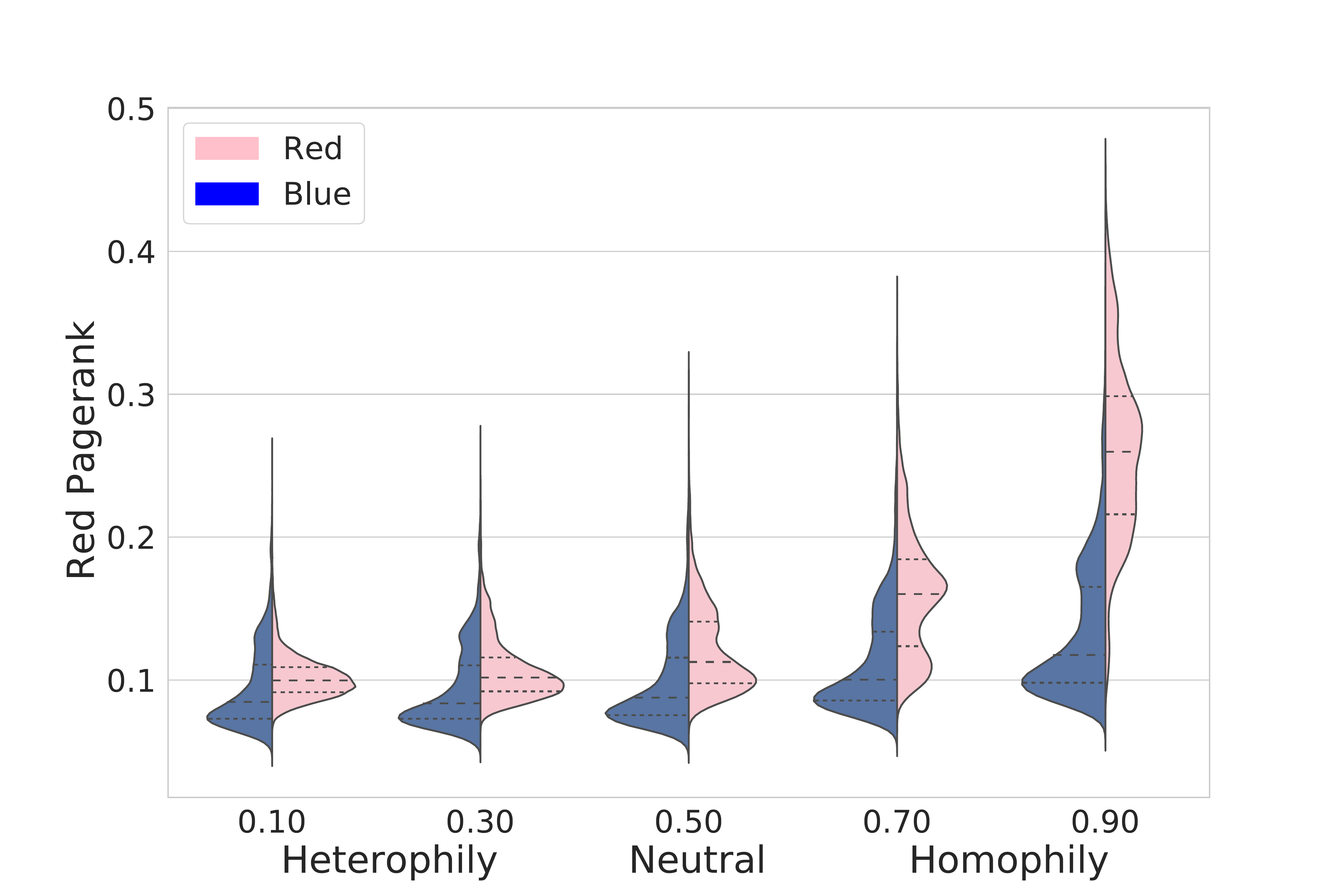}}
	}
	\subfigure[{$r$ = 0.3}]{
		{\includegraphics[width = 0.32\textwidth]{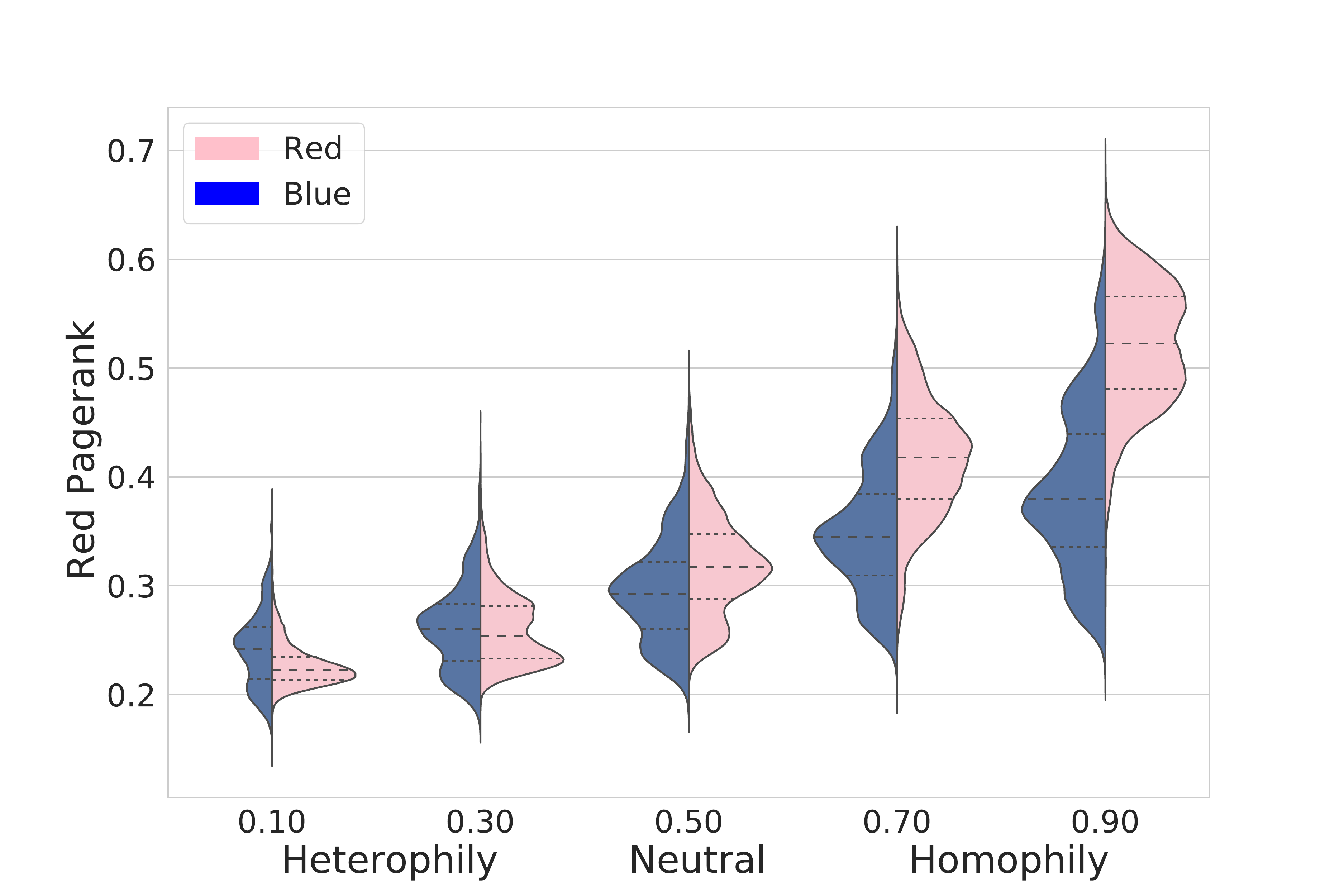}}
	}
	\subfigure[{$r$ =  0.5}]{
		{\includegraphics[width = 0.32\textwidth]{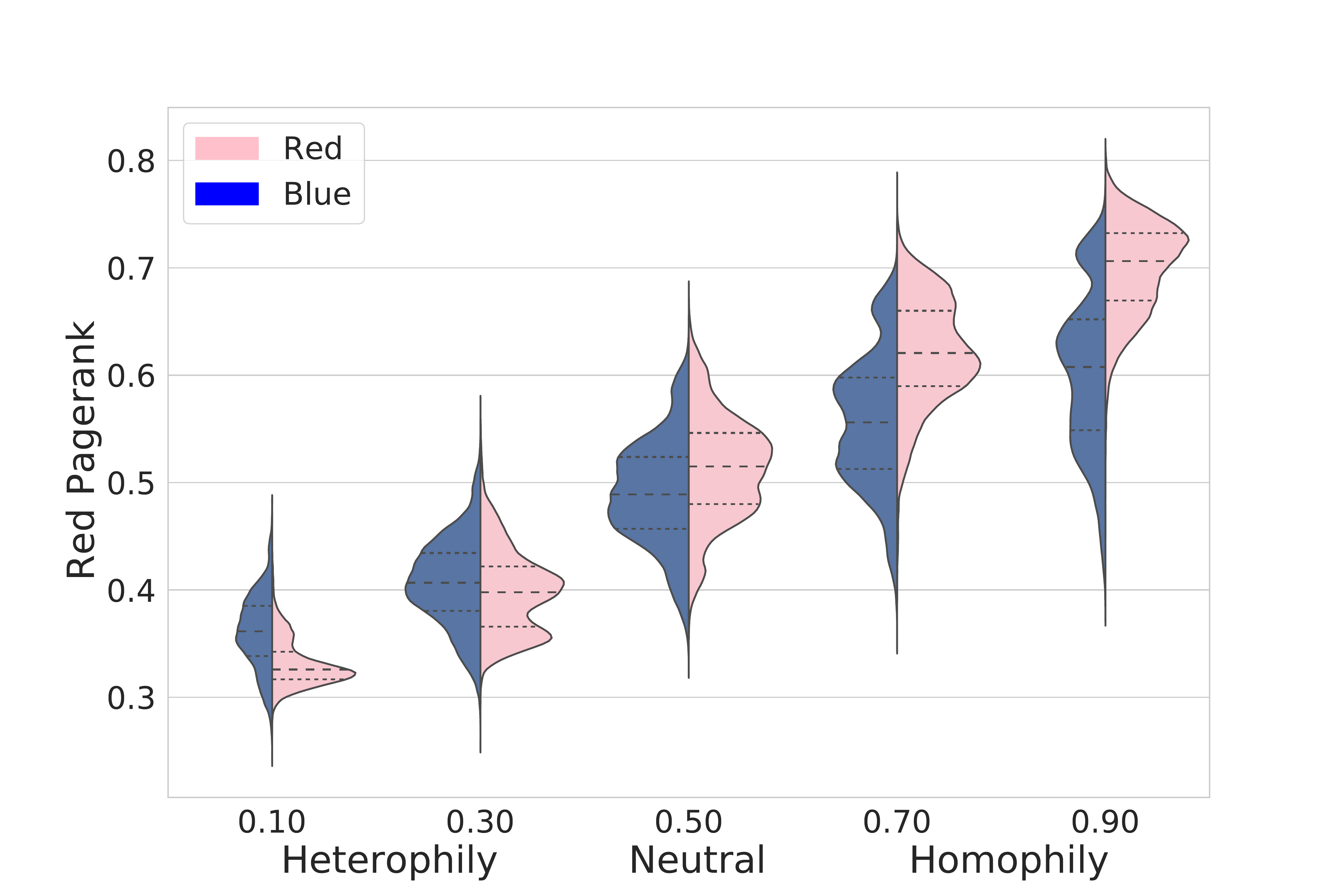}}
	}
	\caption{Distribution of the red personalized pagerank  of the red and blue nodes for varying $\alpha$ in the asymmetric case, $\phi$-fairness achieved when the red pagerank is $r$.}
	\label{fig:asym-violin-plots}
\end{figure*}

\section{Experimental Evaluation}
\label{sec:experiments}
Our goal is to evaluate Pagerank fairness in different kinds of networks, identify the conditions under which Pagerank unfairness emerges and evaluate the effect of the proposed fair Pagerank algorithms. Previous research has shown that homophily and size imbalance may lead to degree unfairness \cite{glass-ceiling,glass-ceiling-recommend,xyz}. Is this the case for Pagerank unfairness? 

Specifically, we address the following three research questions:

\vspace*{0.03in}
\noindent \textbf{RQ1:} Under which conditions are the original Pagerank and personalized Pagerank algorithms fair?

\noindent \textbf{RQ2:} What is the utility loss incurred by the proposed fair Pagerank algorithms in different networks?

\noindent \textbf{RQ3:} What are the qualitative characteristics of the proposed fair Pagerank algorithms?

\begin{figure}[]
	\centering
		{\includegraphics[width = 0.4\textwidth]{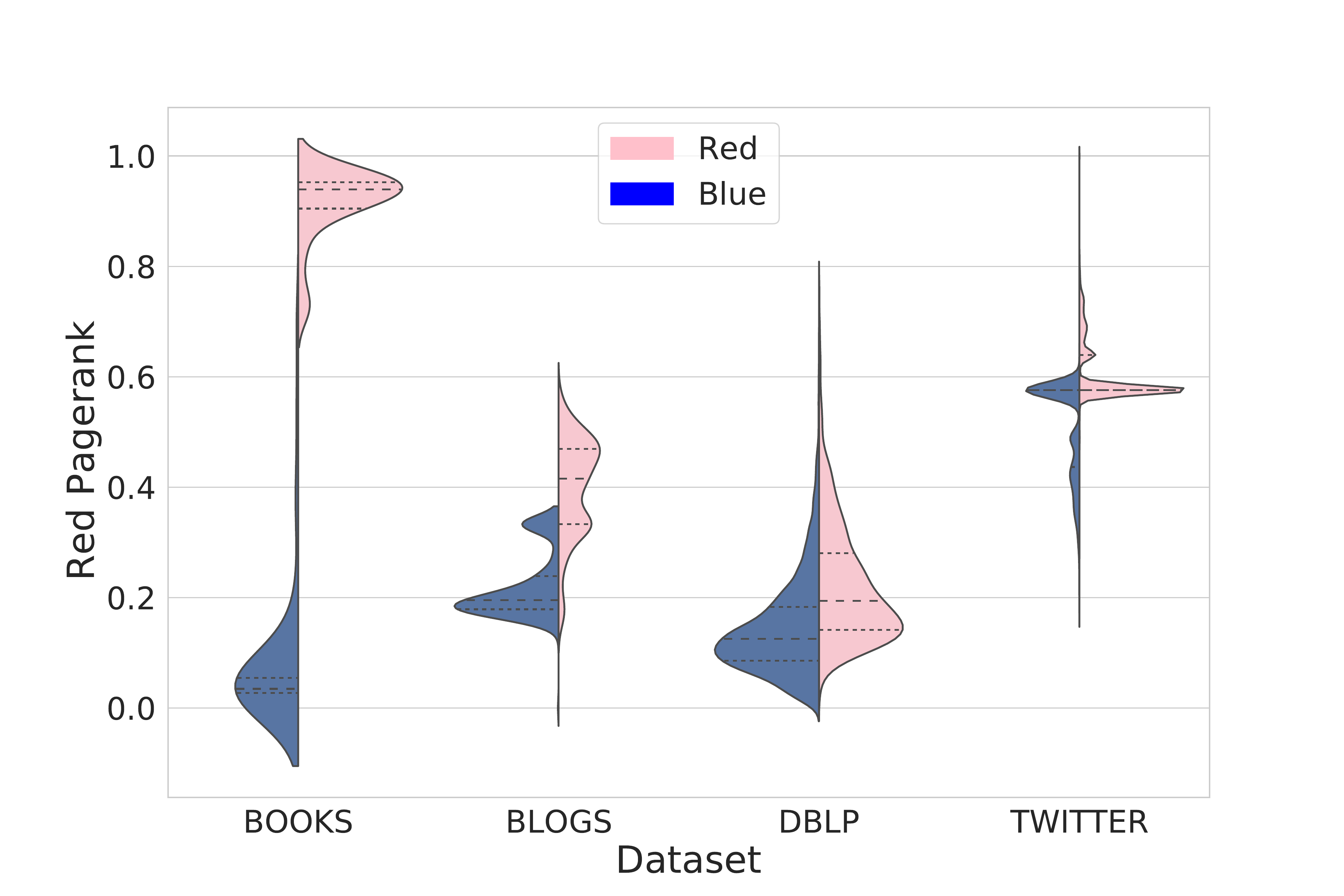}}
	\caption{Distribution of the red personalized pagerank of the red and blue nodes for the real datasets,  $\phi$-fairness when the red pagerank is $r$ (\sc{books} $r$ = 0.47, 	\sc{blogs} $r$ = 0.48, 	\sc{dblp} $r$ = 0.17, and 	\sc{twitter} $r$ = 0.61).}
	\label{fig:real-violin-plots}
\end{figure}

\vspace*{0.03in}
\noindent \textbf{Datasets.}
We use both real and synthetic datasets. Our real datasets are the following:

\begin{itemize}
\item  {\sc{books}}: A network of books about US politics where edges between books represented co-purchasing\footnote{\url{http://www-personal.umich.edu/~mejn/netdata/}}. 
\item {\sc{blogs}}: A directed network of hyperlinks between weblogs on US politic  \cite{blogs-dataset}. 
\item {\sc{dblp}}: An author collaboration network constructed from DBLP including a subset of data mining and database conferences.
\item {\sc{twitter}}: A political retweet graph from \cite{nr}. 
\end{itemize}

The characteristics of the real datasets are summarized in Table \ref{table:real}.
To infer the gender in  {\sc{dblp}}, we used   the python \textit{gender guesser} package\footnote{\url{https://pypi.org/project/gender-guesser/}}.
Regarding $homophily$,
we measure for each group, the percentage of the edges that are \textit{cross-edges} that is they point to nodes of the other group divided by the expected
number of such edges. We denote these quantities as $\textit{cross}_R$ and $\textit{cross}_B$.
Values  significantly smaller than 1 indicate that the corresponding group exhibits homophily \cite{network-book}. 

Synthetic networks are generated using the biased preferential attachment model introduced in \cite{glass-ceiling}. The graph evolves over time as follows.
Let  $G_t = (V_t, E_t)$ and $d_t(v)$ denote the graph and the degree of node $v$ at time $t$, respectively.
The process starts with an arbitrary connected graph $G_0$, with $n_0 \, r$ red and $n_0 \,(1 - r)$ blue nodes.
At  time step $t + 1$, $t > 0$, a new node $v$ enters the graph. The color of $v$ is red with probability $r$ and blue with probability $1-r$. Node $v$ chooses to connect with an existing node $u$ with probability
$\frac{d_t(u)}{\sum_{w \in G_{t} d_t(w)}}$. If the color of the chosen node $u$ is the same with the color of the new node $v$, then an edge between them is inserted with probability $\alpha$; otherwise an edge is inserted with probability
$1-\alpha$. If no edge is inserted, the
process of selecting a neighbor for node $v$ is repeated until an edge is created.

\begin{figure*}[]
	\centering
\subfigure[{symmetric, $r$ = 0.3}]{
	{\includegraphics[width = 0.22\textwidth]{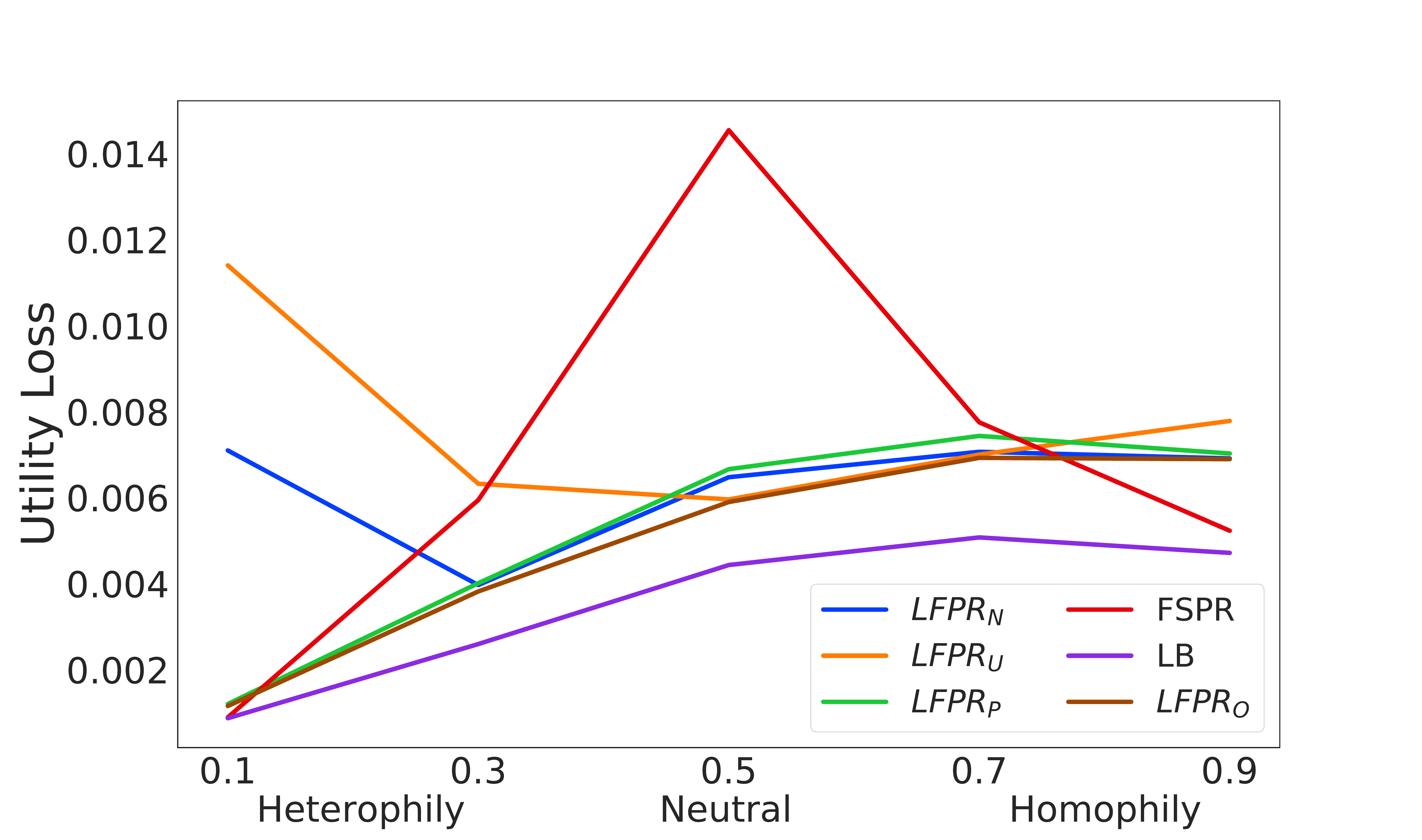}}
}
	\subfigure[{asymmetric, $r$ =  0.3}]{
		{\includegraphics[width=0.22\textwidth]{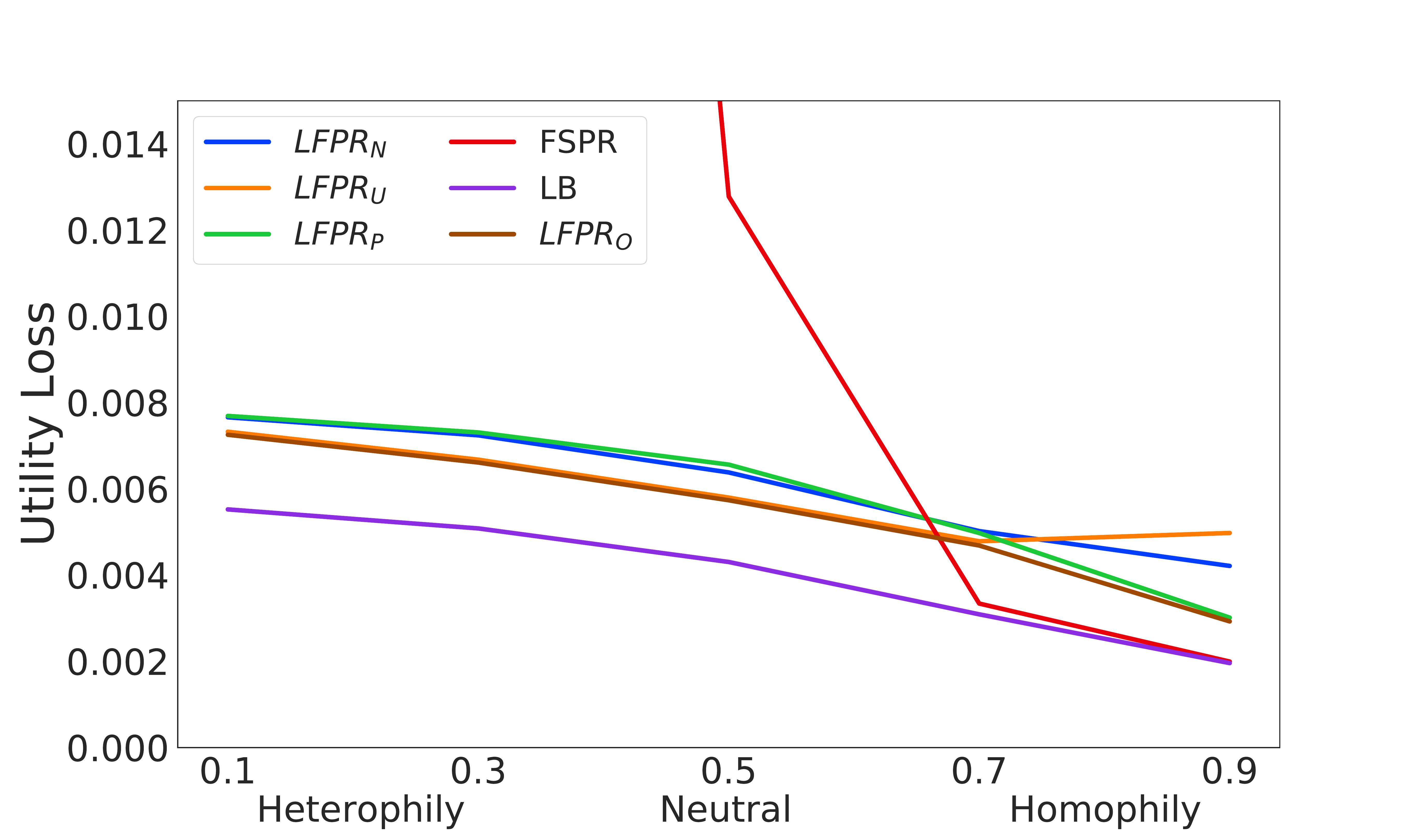}}
	}
\subfigure[{symmetric, $r$ = 0.5}]{
	{\includegraphics[width = 0.22\textwidth]{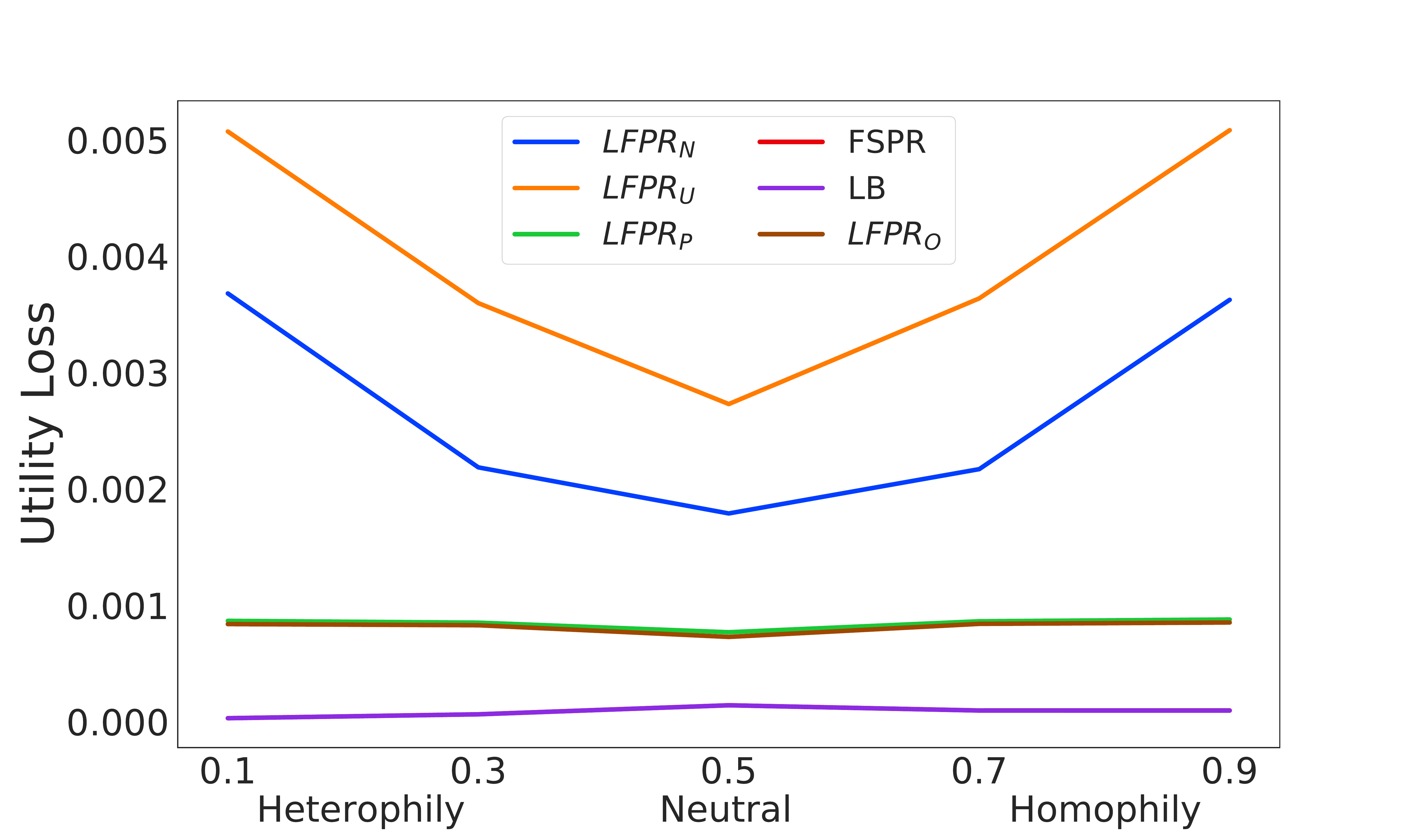}}
}
\subfigure[{asymmetric, $r$ =  0.5}]{
	{\includegraphics[width=0.22\textwidth]{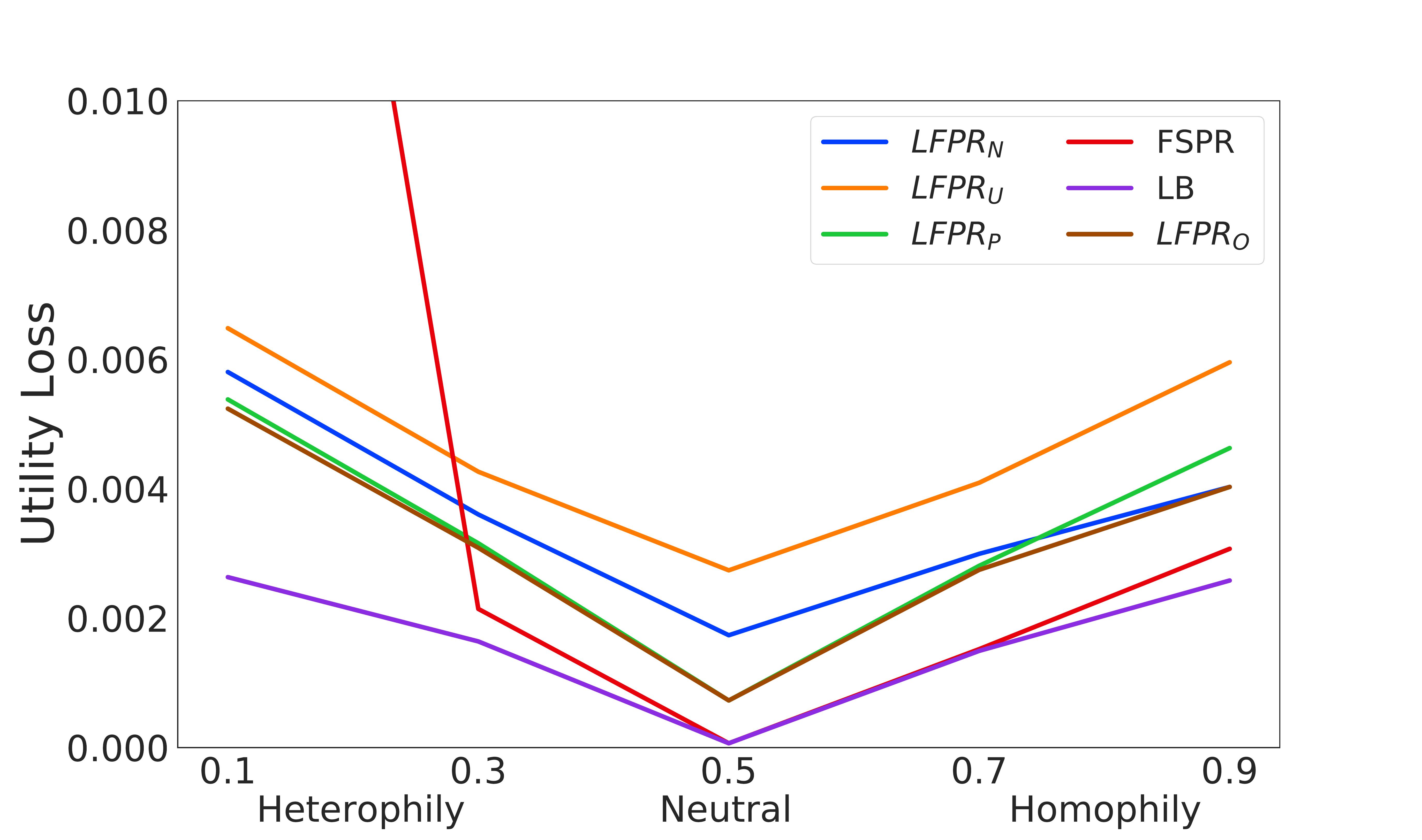}}
}
	\caption{Utility loss for synthetic networks, $\phi$ = 0.5.}
	\label{fig:synth-utility-homo}
\end{figure*}

\begin{figure}[]
	\centering
	\subfigure[{$r$ = 0.3}]{
		{\includegraphics[width = 0.22\textwidth]{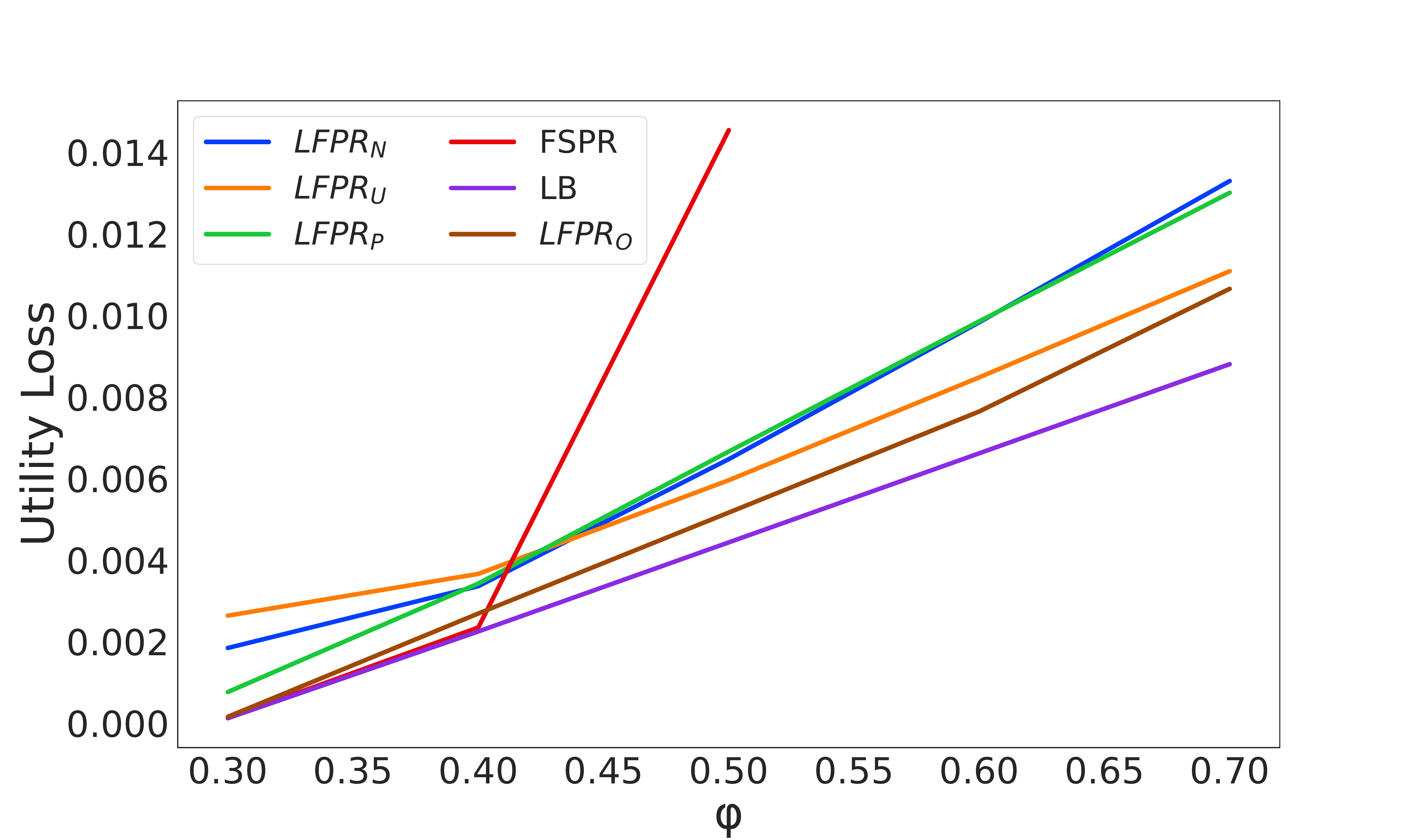}}
	}
	\subfigure[{$r$ = 0.5}]{
		{\includegraphics[width = 0.22\textwidth]{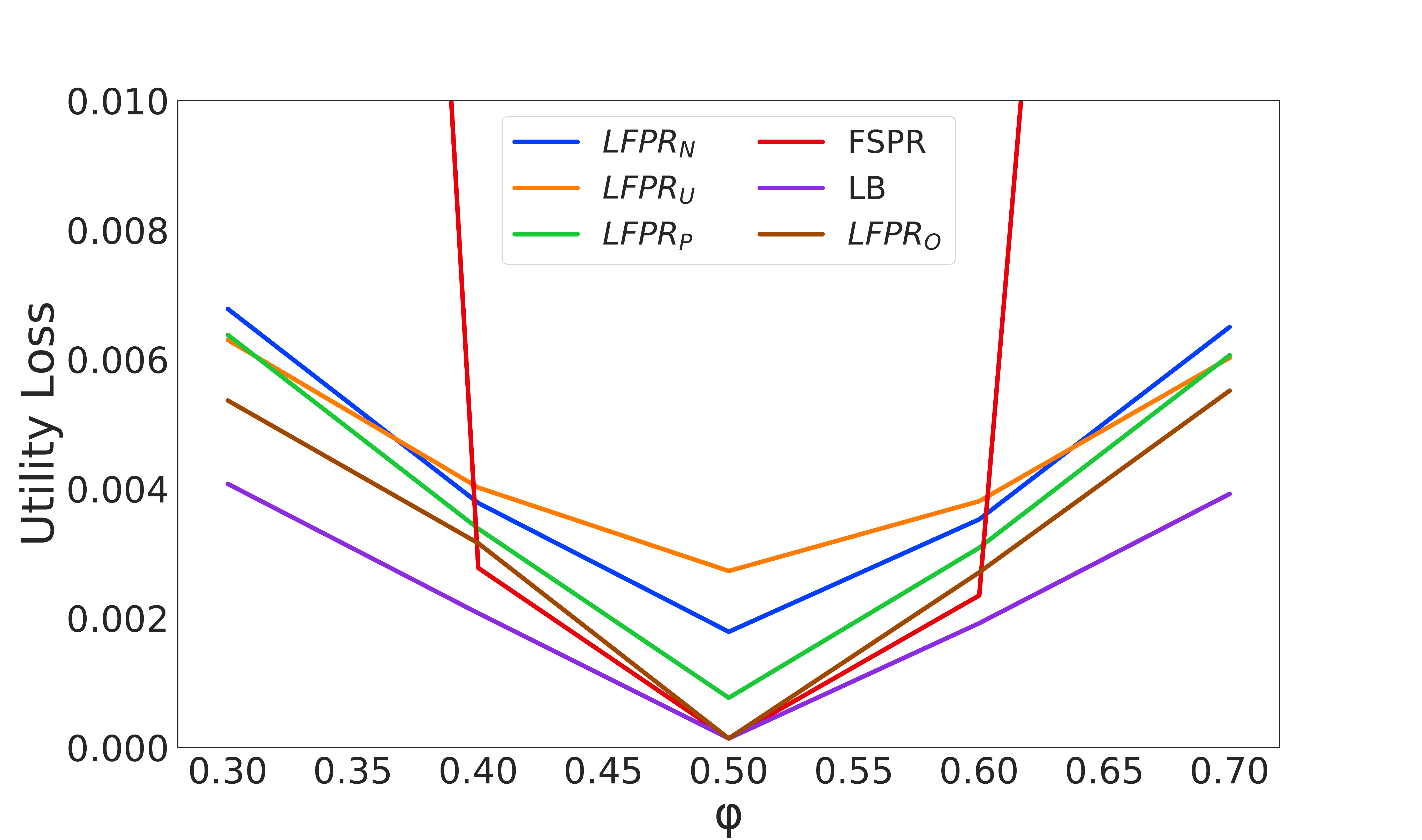}}
	}
	\caption{Utility loss for the synthetic datasets, $\alpha$ = 0.5.}
	\label{fig:synth-utility-phi}
\end{figure}

\begin{figure*}[]
	\centering
	\subfigure[{\sc{books}}]{
		{\includegraphics[width = 0.22\textwidth]{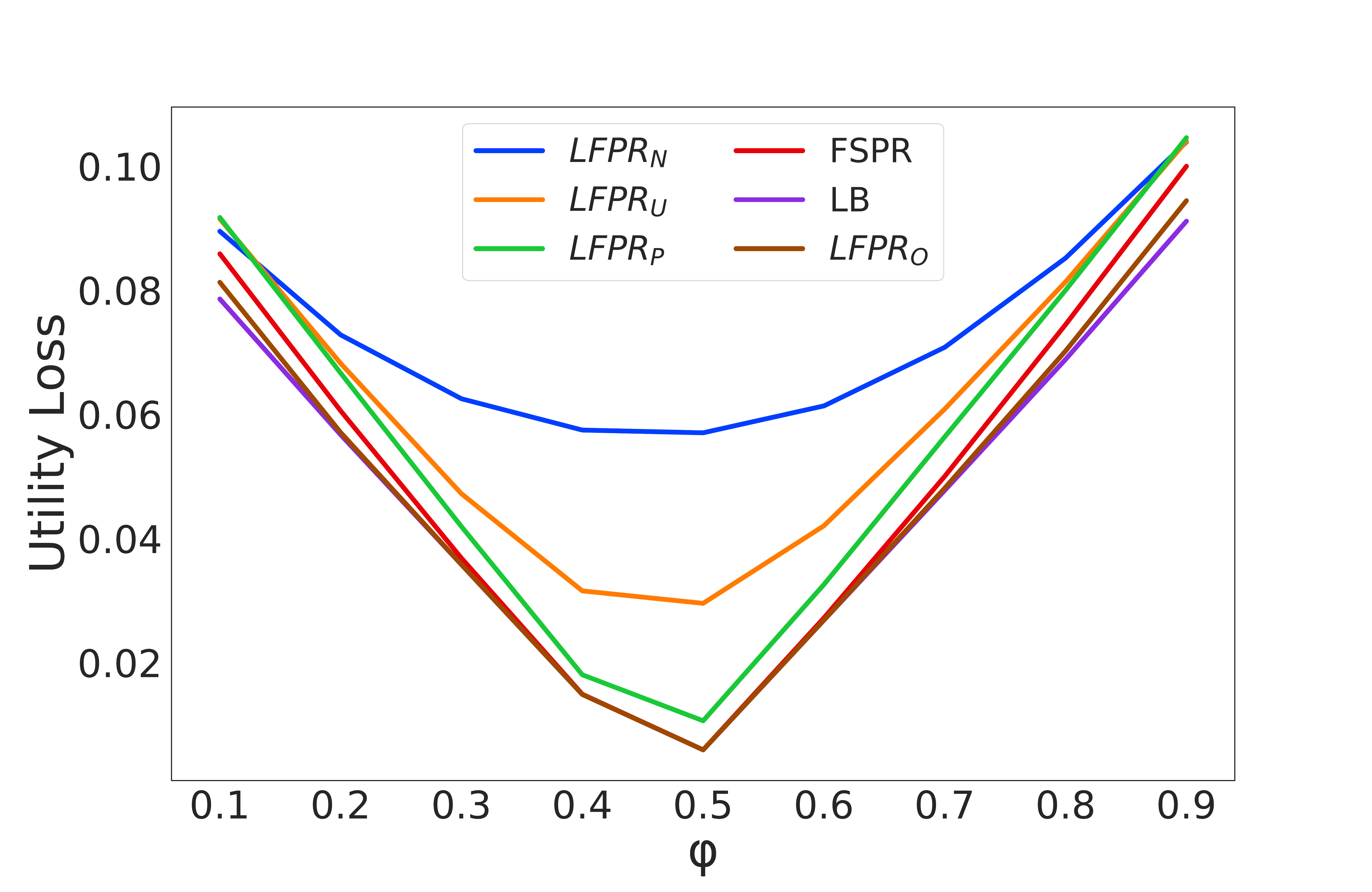}}
	}
	\subfigure[{\sc{blogs}}]{
		{\includegraphics[width = 0.22\textwidth]{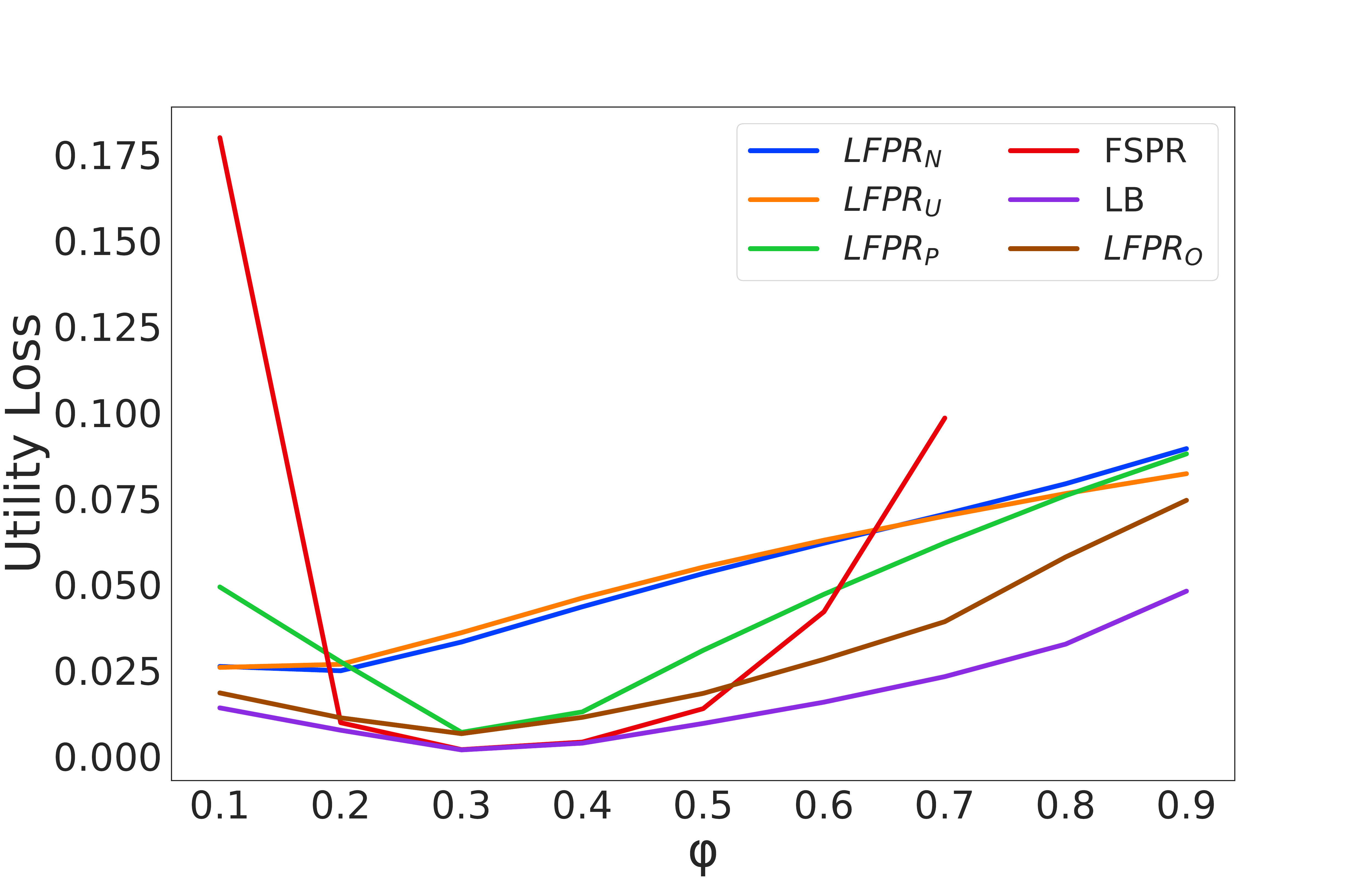}}
	}
	\subfigure[{\sc{dblp}}]{
		{\includegraphics[width = 0.22\textwidth]{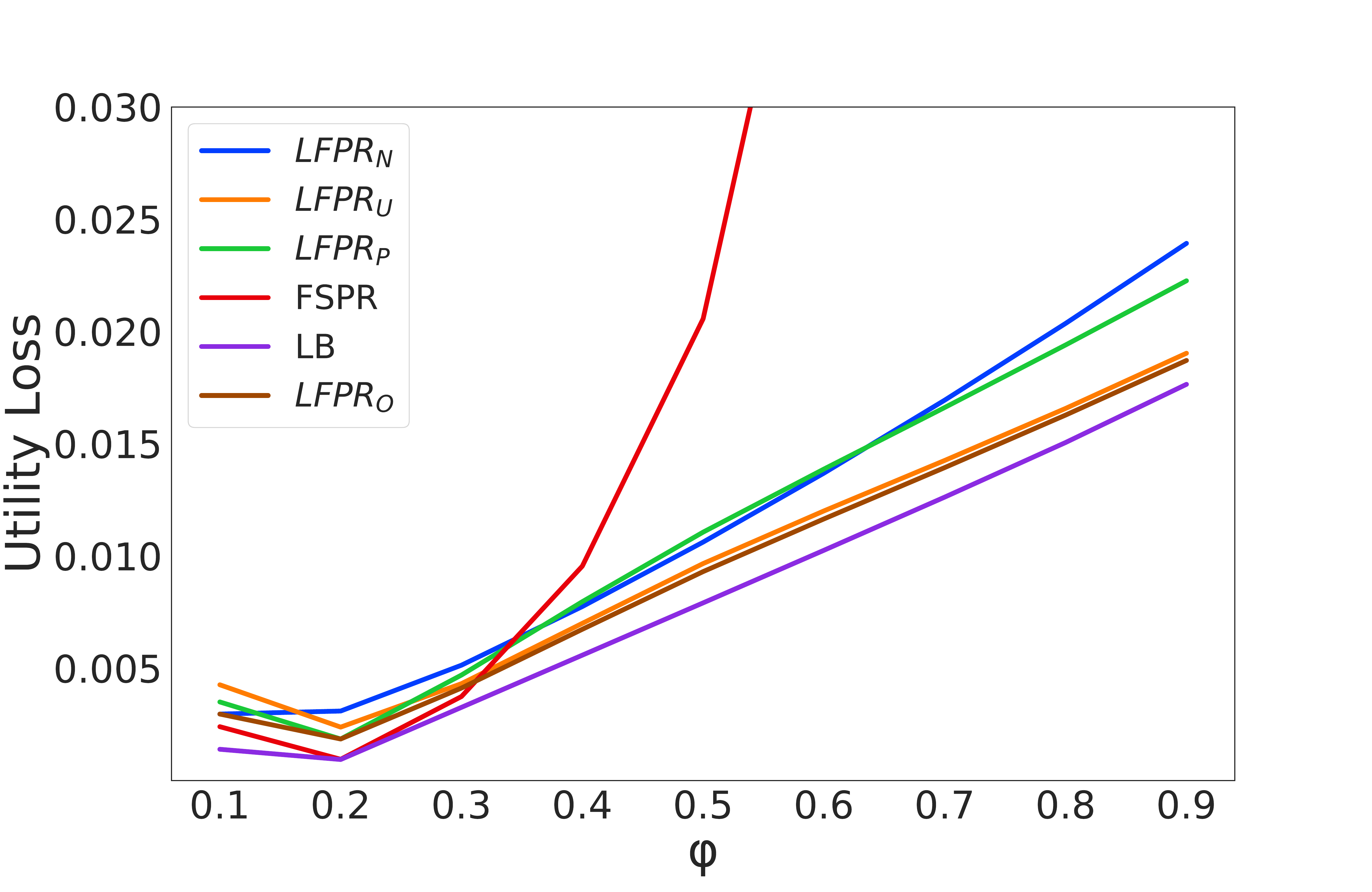}}
	}
	\subfigure[{\sc{twitter}}]{
		{\includegraphics[width = 0.22\textwidth]{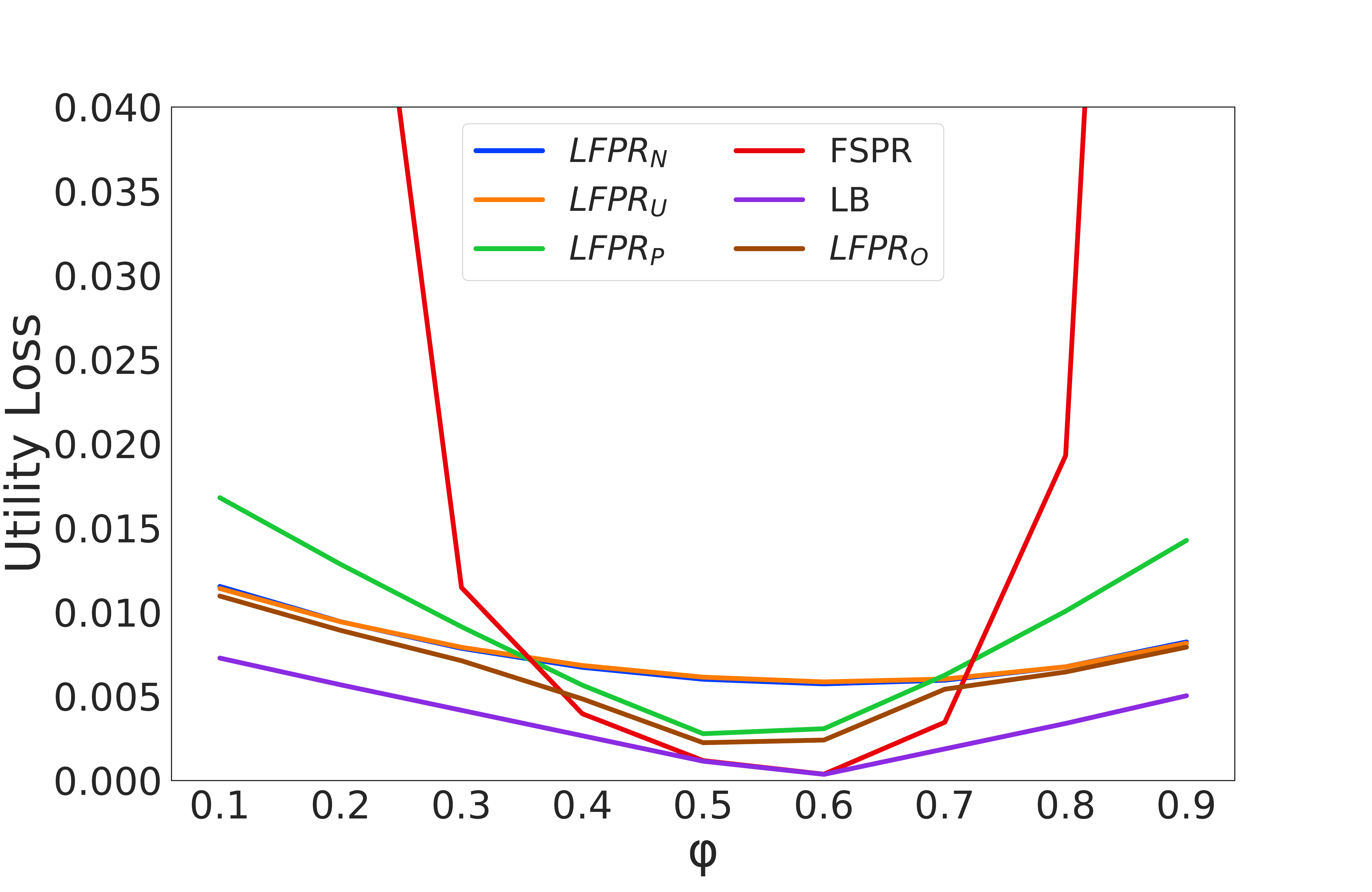}}
	}
	\caption{Utility loss for the real networks, (original red pagerank, \sc{books}: 0.48, \sc{blogs}: 0.33, \sc{dblp}: 0.16, and \sc{twitter}: 0.57).}
	\label{fig:real-utility}
\end{figure*}

Parameter $r$ controls the group size imbalance.
Parameter $\alpha$ controls the level of homophily: $\alpha < 0.5$ corresponds to heterophily,
$\alpha = 0.5$ to neutrality and $\alpha > 0.5$ to homophily.
We consider: (a) a symmetric case, where $\alpha$ is the same for both groups and (b)
an asymmetric case, where we set $\alpha$ = 0.5 for the blue group, making it neutral, and vary $\alpha$ for the red group.

The datasets and code are available in GitHub\footnote{\url{https://github.com/SotirisTsioutsiouliklis/FairLaR}}.

\subsection{When is Pagerank Fair?}
We study the conditions under which the original Pagerank and personalized Pagerank
algorithms are fair. We assume that the algorithms are fair, if they respect demographic parity, that is, if each group gets pagerank equal to its ratio in the overall population ($\phi$ = $r$). For brevity, we shall call the (personalized) pagerank allocated to red  nodes
\textit{red (personalized) pagerank} and the (personalized) pagerank allocated to blue  nodes \textit{blue (personalized) pagerank} .

First, we study homophily and size imbalance using synthetic datasets.
In Figure \ref{fig:local-synthetic-all}(a), we report the  red pagerank  
for the symmetric and in Figure \ref{fig:local-synthetic-all}(b) for the asymmetric case.
Fairness corresponds to the identity line (red pagerank = $r$).
Values above the line indicate unfairness towards the blue group, while values below the line unfairness towards the red group.

We also plot the distribution of the red personalized pagerank 
in Figures \ref{fig:sym-violin-plots} and  \ref{fig:asym-violin-plots}
for the symmetric and asymmetric case respectively.
To test whether the red personalized pagerank of a node depends on its color, 
we plot two distributions, one for the red personalized pagerank of the red nodes and one for the red personalized pagerank of the blue nodes.
Distributions are plotted in the form of violin plots. Personalized pagerank fairness corresponds to the case in which the two distributions overlap, with their mean on value $r$. 
Note that when a group is homophilic, it is expected that a large part of its personalized pagerannk goes to nodes of its own color, e.g., red homophilic nodes have large red personalized pageranks. 

\noindent \textbf {Symmetric case} (Figures \ref{fig:local-synthetic-all}(a) and \ref{fig:sym-violin-plots}): When both groups are  homophilic ($\alpha$ = 0.7, 0.9), the nodes tend to form two clusters, one with red nodes and one with blue nodes
sparsely connected to each other. 
This leads to an almost fair pagerank (with a very slight unfairness towards the minority group), but highly unfair personalized pageranks. 
On the contrary, when there is 
 heterophily ($\alpha = 0.1, 0.3$),
 there are no clusters, nodes tend to connect with nodes of the opposite color, and the larger group favors the smaller one. In this case, the pagerank and personalized pageranks of both the blue and the red nodes are all unfair towards the majority group. This is  especially evident when the imbalance in size is large (small $r$). 
 
\noindent \textbf {Asymmetric case} (Figures \ref{fig:local-synthetic-all}(b) and \ref{fig:asym-violin-plots}):
When the red nodes are homophilic ($\alpha$ = 0.7, 0.9), the red group keeps the pagerank to itself. As a result both pagerank and personalized pageranks are unfair towards the blue nodes, especially for larger $r$.
When the red nodes are heterophilic ($\alpha$ = 0.1, 0.3), the red group favors the blue group, and as a result, both the pagerank and the personalized pagerank are unfair towards the red nodes, especially for larger $r$. 
Thus, independently of the size $r$, 
pagerank is unfair to the  blue (the neutral) group in case of homophily, and unfair to the red group in the case of heterophily. 

\noindent \textbf {Universal fairness:} The only case when both pagerank and  personalized pageranks  are all fair is in a neutral network ($\alpha =  0.5$) with same-size groups ($r$ = 0.5) (middle violin plots in Figures 
\ref{fig:sym-violin-plots}(c) and \ref{fig:asym-violin-plots}(c)).

\noindent \textbf {Real datasets:}
For the real datasets, we report the red pagerank in Table \ref{table:real} ($p_R$ value) and plot the distributions of the red personalized pagerank of the blue and red nodes in Figure \ref{fig:real-violin-plots}.
For {\sc books}, there is no size imbalance and there is strong symmetric homophily leading to fair pagerank and highly unfair personalized pageranks.
For {\sc blogs}, there is no size imbalance, but the blue group is slightly more homophilic, which leads both to unfairness in pagerank for the red group, and unfairness of the personalized pagerank of the blue nodes towards the red ones. 
For {\sc dblp}, we have large size imbalance with the red group being the minority but almost neutral behavior in terms of homophily, leading to an almost fair pagepank, and the majority of  personalized pageranks being fair. 
Finally, for {\sc twitter}, the red group is larger than the blue group but less homophilic which leads to a slight unfairness towards the red group for both pagerank and personalized pageranks.

\subsection{What is the Utility Loss for Fairness?}
We now look into the utility loss
for achieving fairness.
We can view utility loss for each network as a measure of the cost we have to pay to achieve $\phi$-fairness for this network.
First, 
to assess the utility loss of our algorithms in absolute terms 
we compute a lower bound for the utility loss.

\noindent \textbf{Lower Bound.} 
We  compute a lower bound on the utility loss, by constructing the probability vector $\wvec$ that is $\phi$-fair, and it has the minimum utility loss compared to the original pagerank vector $\pvec_O$. 
Note that vector $\wvec$ is not necessarily attainable by any Pagerank algorithm in $\allPR$.

To compute $\wvec$, we start with $\pvec_O$ and we redistribute the probability between the two groups to make it fair. Let $\pvec_O(R)$ be the probability assigned to the red group by $\pvec$. Without loss of generality, assume that $\pvec_O(R)< \phi$, and let $\Delta = \phi-\pvec_O(R)$. To make the vector fair, we need to remove  $\Delta$ probability mass from the nodes in $B$, and redistribute it to the nodes in $R$. It is easy to show that to minimize the loss, the optimal redistribution will remove uniformly $\Delta/|B|$ probability from all nodes in $B$, and add uniformly $\Delta/|R|$ to all nodes in $R$. 
This follows from the fact that among all distribution vectors the one with the smallest length is the uniform one. However, this process does not guarantee that the resulting vector will not have negative entries, since some blue nodes may have probability less than $\Delta/|B|$. Let $\beta$ be the smallest non-zero such probability of any blue node. Our algorithm transfers $\beta$ probability from all the non-zero blue nodes to the red nodes, and then recursively applies the same procedure for the residual amount of probability that has not been transferred. It is not hard to show that this process will produce a fair vector with the minimum utility loss with respect to $\pvec_O$.

Figures \ref{fig:synth-utility-homo} and \ref{fig:synth-utility-phi} report the utility loss for selected synthetic networks 
and Figure \ref{fig:real-utility} for different values of
$\phi$ for the real datasets. $LB$ is the 
lower bound on utility loss.

\noindent \textbf{Effect of $\phi$:} In all cases, loss increases as the requested $\phi$ deviates from the red pagerank originally assigned to the red group
(Figures \ref{fig:synth-utility-phi} and  \ref{fig:real-utility}).

\noindent \textbf{\textsl{FSPR}:} \, In some cases {\sensitive} incurs high utility loss and, occasionally, it is even unable to find an appropriate solution. 
{\sensitive} achieves fairness by changing the jump vector of the Pagerank algorithm. 
The overall pagerank vector is a linear combination of the personalized pagerank vectors, with the jump vector providing the coefficients of this linear combination.  
{\sensitive} increases the jump probability for the vectors that favor the group it wants to promote and takes away probability from the vectors that favor the other group. However, when there are few, or no appropriate such vectors, {\sensitive} is forced to make extreme choices (assign a lot of probability to a few vectors) thus incurring high utility loss, or it is unable to find any solution.

There are several such examples in our datasets. 
For the {\sc dblp} dataset (Figure \ref{fig:real-utility}(c)), for small values of $\phi$ ($\phi \leq 0.3$), the utility loss of {\sensitive} is close to the optimal, but for $\phi \geq 0.4$, it skyrockets. Looking at Figure \ref{fig:real-violin-plots}, we see that there are very few personalized pagerank vectors with red pagerank larger than 0.4. As a result, for  $\phi \geq 0.4$, {\sensitive} is forced to allocate all the probability of the jump vector to these nodes, leading to high utility loss.
It is also interesting to observe the effect of homophily, or lack of, on the utility loss of {\sensitive}. In Figure \ref{fig:synth-utility-homo}(a) and (b), utility loss peaks when the network is neutral ($\alpha = 0.5$). In this case, there is no personalized pagerank vector that strongly favors one group over the other.

\noindent \textbf{\textsl{LFPR}:} \, Overall, for the locally fair family of algorithms, the utility loss function is smoother, avoiding high peaks.
The locally fair algorithms are directly affected by homophily, since this determines the composition of the neighborhoods of the nodes. As we deviate from neutrality, the loss increases (Figure	\ref{fig:synth-utility-homo}). 
This holds especially for the {\neighborlocal} algorithm.
This can be seen very clearly in  {\sc books} (Figure \ref{fig:real-utility}(a)),  where {\sensitive} almost achieves the lower bound, while {\neighborlocal} incurs high utility loss because of {\sc books} being very homophilic.
The utility loss of {\uniformlocal} and {\proportionallocal} follows in general the same trend as {\neighborlocal}.
Finally, {\residuallocal} redistributes any residual Pagerank so that the
utility loss is optimized and in many cases its utility loss is close to the lower bound ($LB$). 

\noindent \textbf{Summary:} 
{\sensitive} works well when there are enough unfair nodes (i.e., nodes with unfair personalized pageranks), as it can distribute the jump probability to them to achieve fairness.  
On the contrary, locally fair algorithms have high utility loss when there are many unfair nodes. {\neighborlocal} is the most sensitive to homophily.
The utility loss of {\neighborlocal} can be seen as a measure of local unfairness.
Overall, the locally fair algorithms are more stable than {\sensitive}. {\residuallocal} works very well in terms of utility loss and in many cases it achieves utility close to the lower bound.

\subsection{Qualitative Evaluation}
In this section, we provide some qualitative experiments, to better understand the properties of the fair algorithms.

\begin{figure}[t]
	\centering
\subfigure[Original {\pagerank}]{
	{\includegraphics[width = 0.175\textwidth]{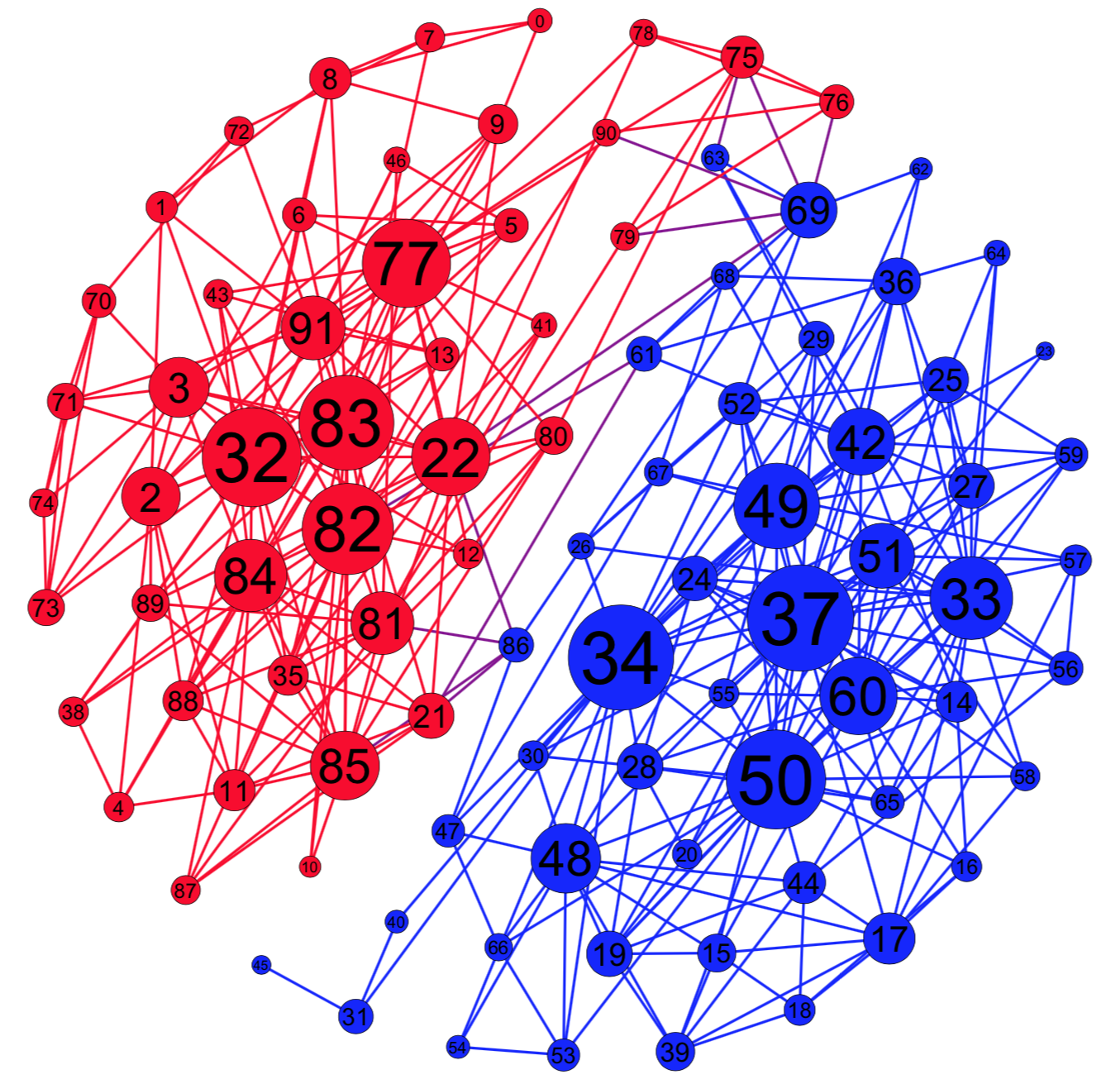}}
}
\subfigure[{\neighborlocal}]{
	{\includegraphics[width = 0.175\textwidth]{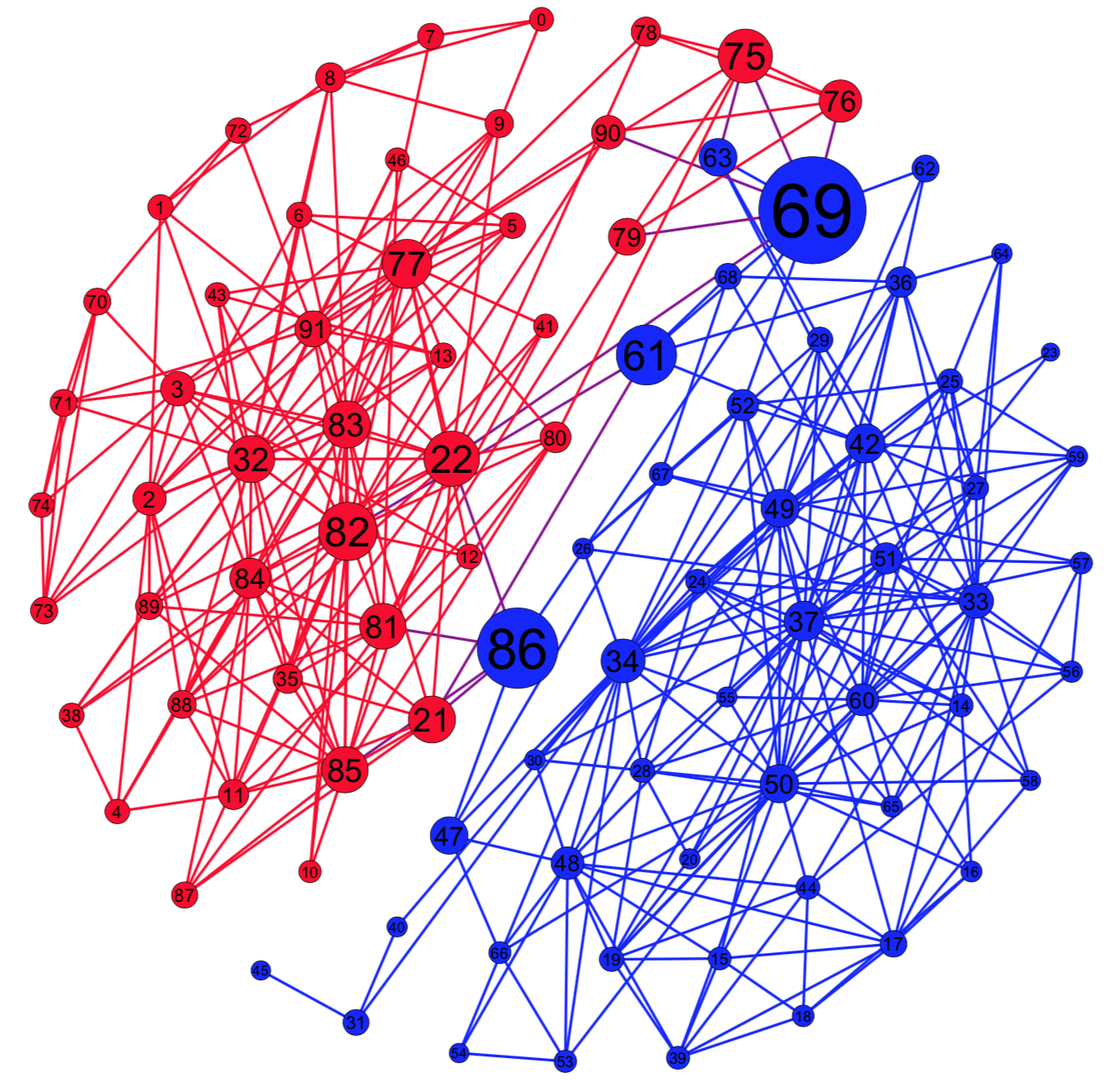}}
}
\subfigure[{\sensitive}]{
	{\includegraphics[width = 0.175\textwidth]{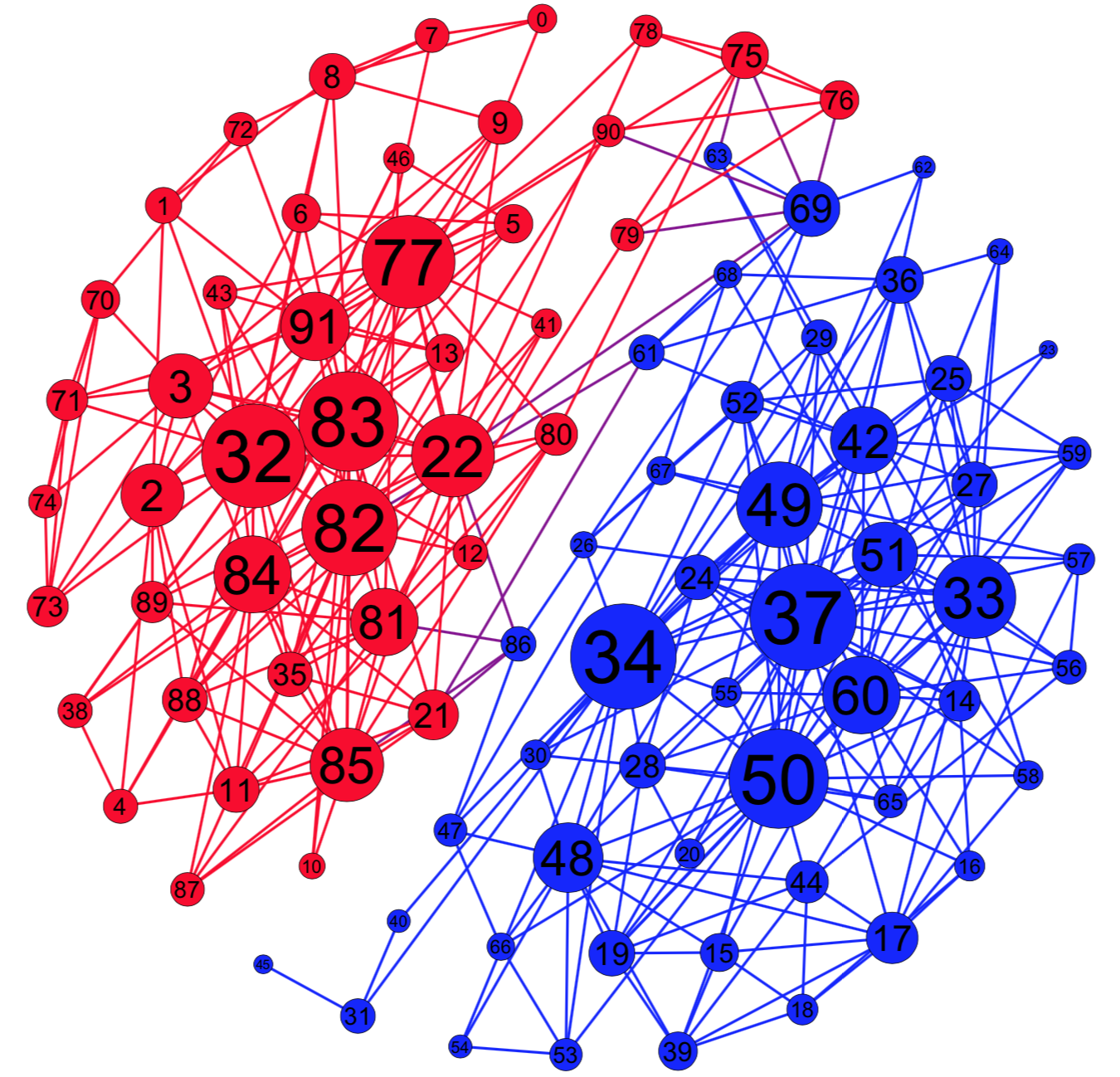}}
}
\subfigure[Jump vector for {\sensitive}]{
	{\includegraphics[width = 0.175\textwidth]{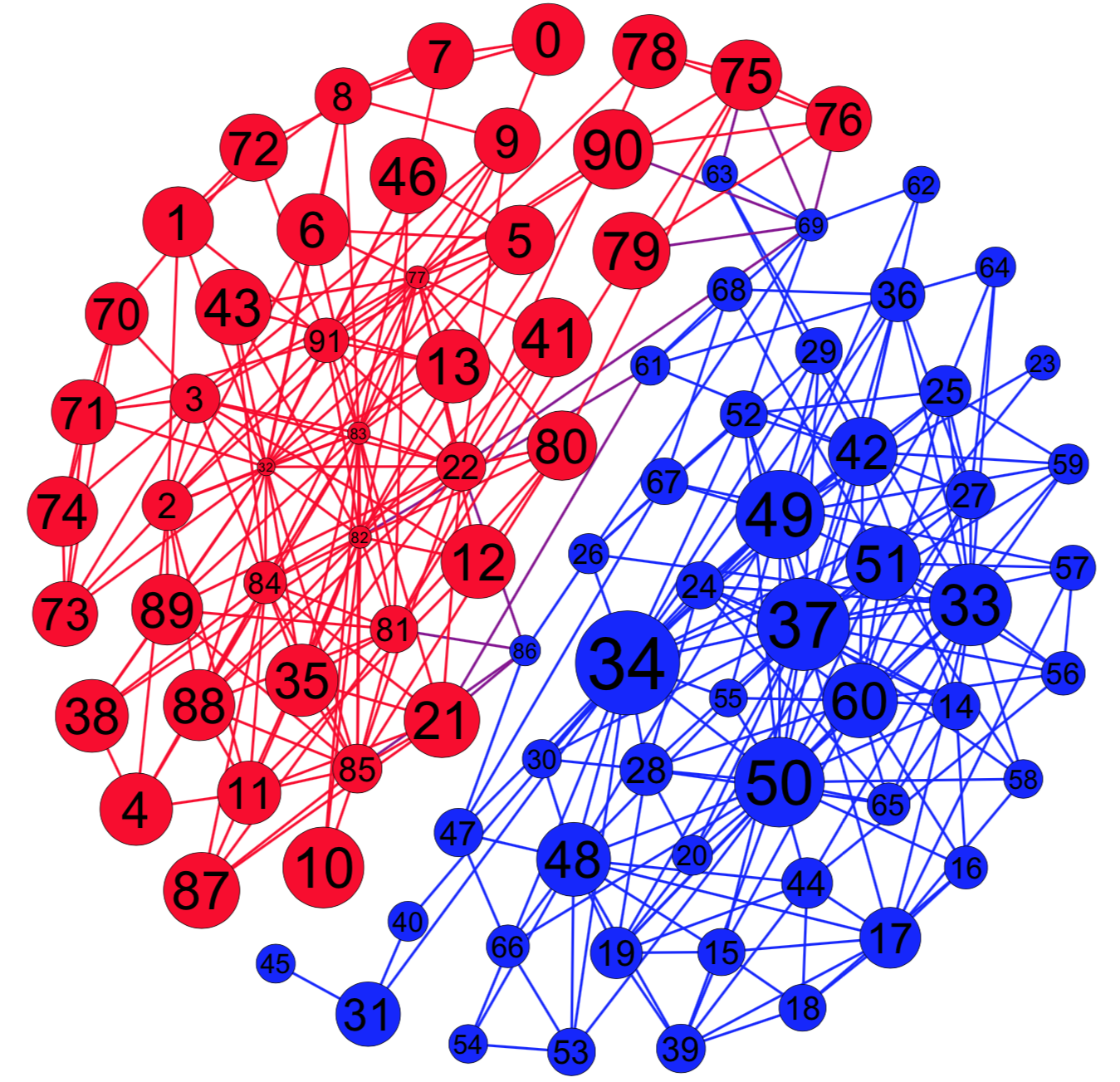}}
}
\caption{Visualization of the {\sc book} dataset.}
\label{fig:vis-books}
\end{figure}

\begin{figure}[t]
	\centering
	\subfigure[Original {\pagerank}]{
		{\includegraphics[width = 0.22\textwidth]{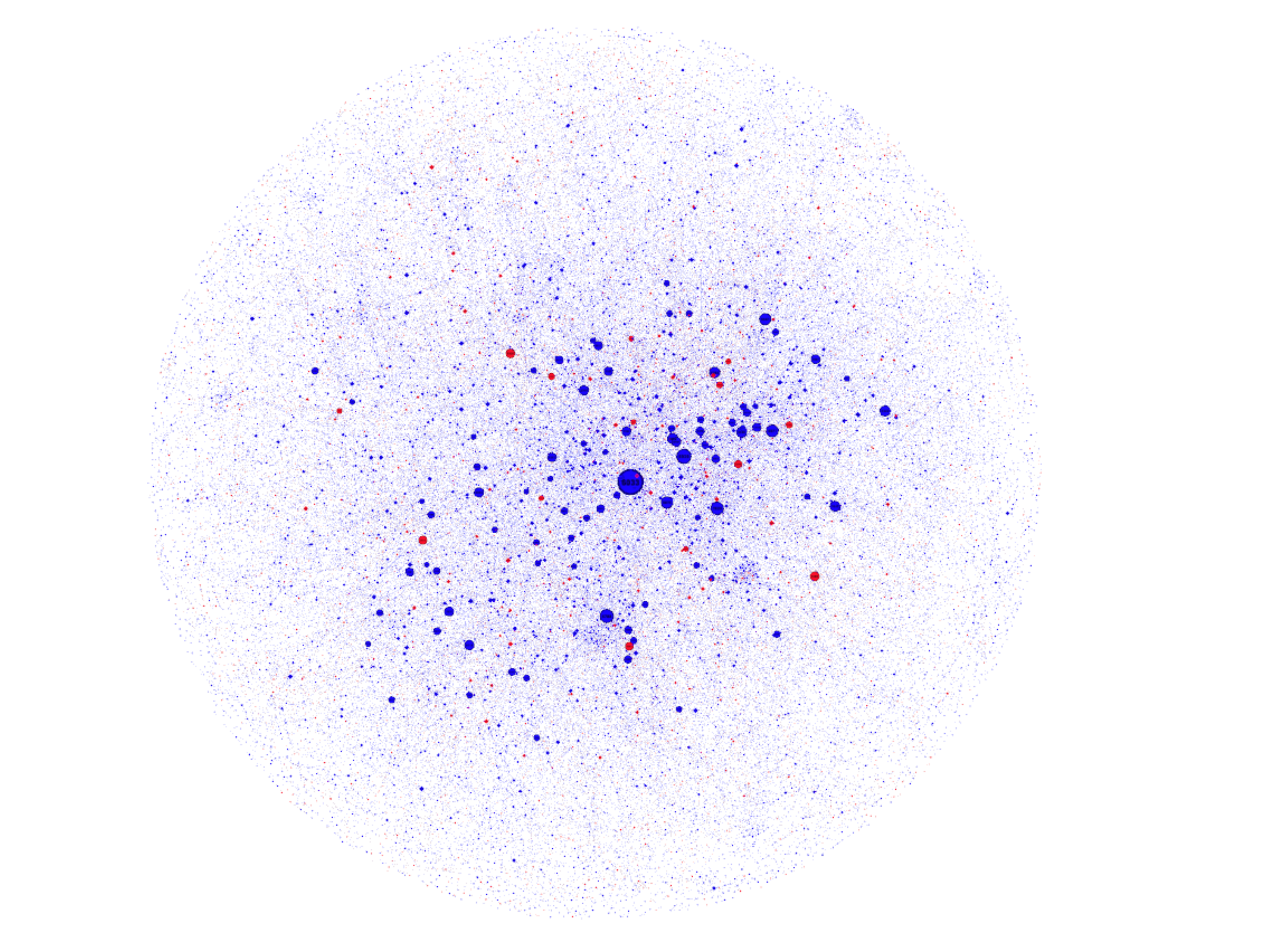}}
	}
	\subfigure[{\neighborlocal}]{
	{\includegraphics[width = 0.22\textwidth]{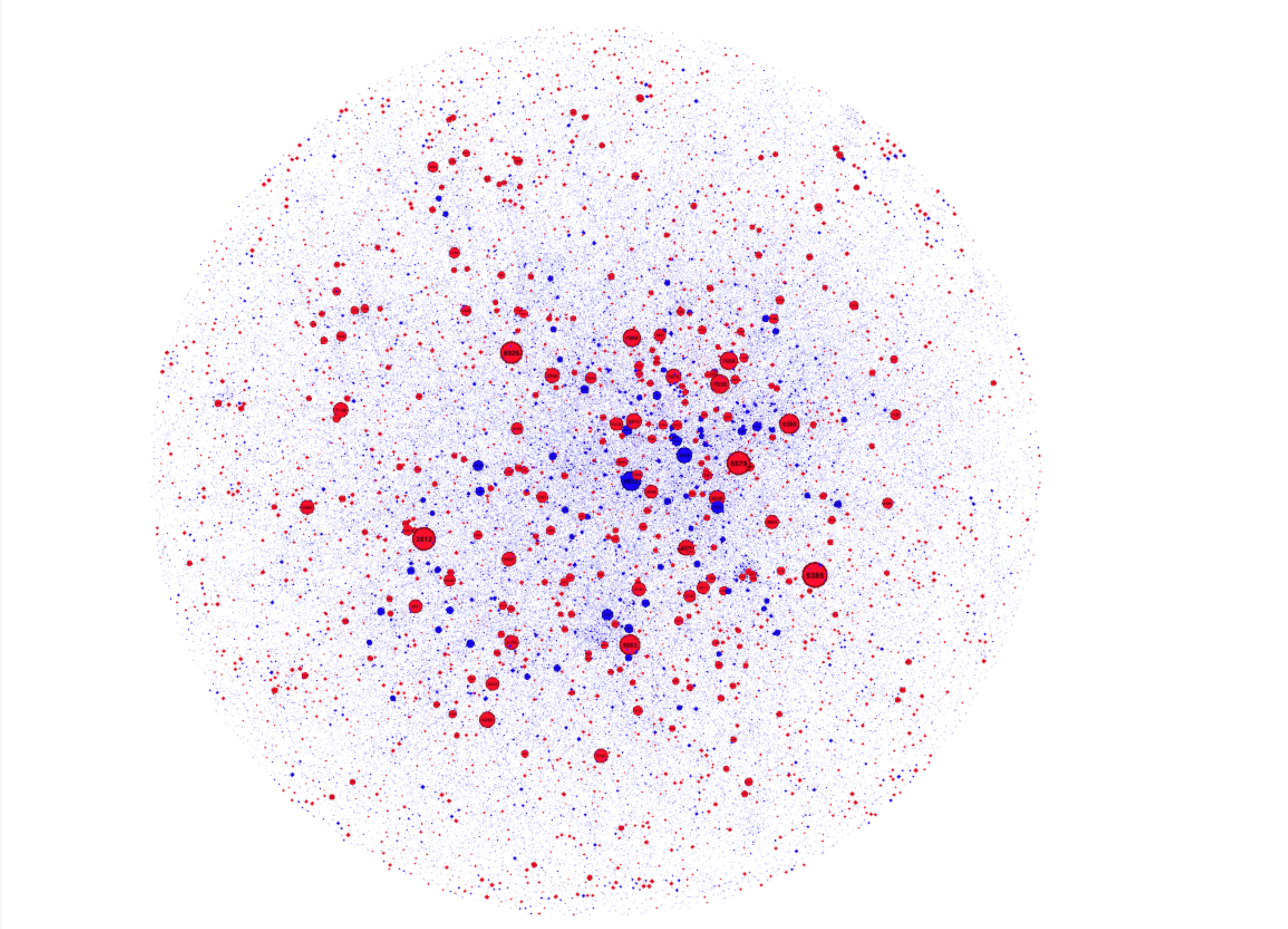}}
	}
	\subfigure[{\sensitive}]{
		{\includegraphics[width = 0.22\textwidth]{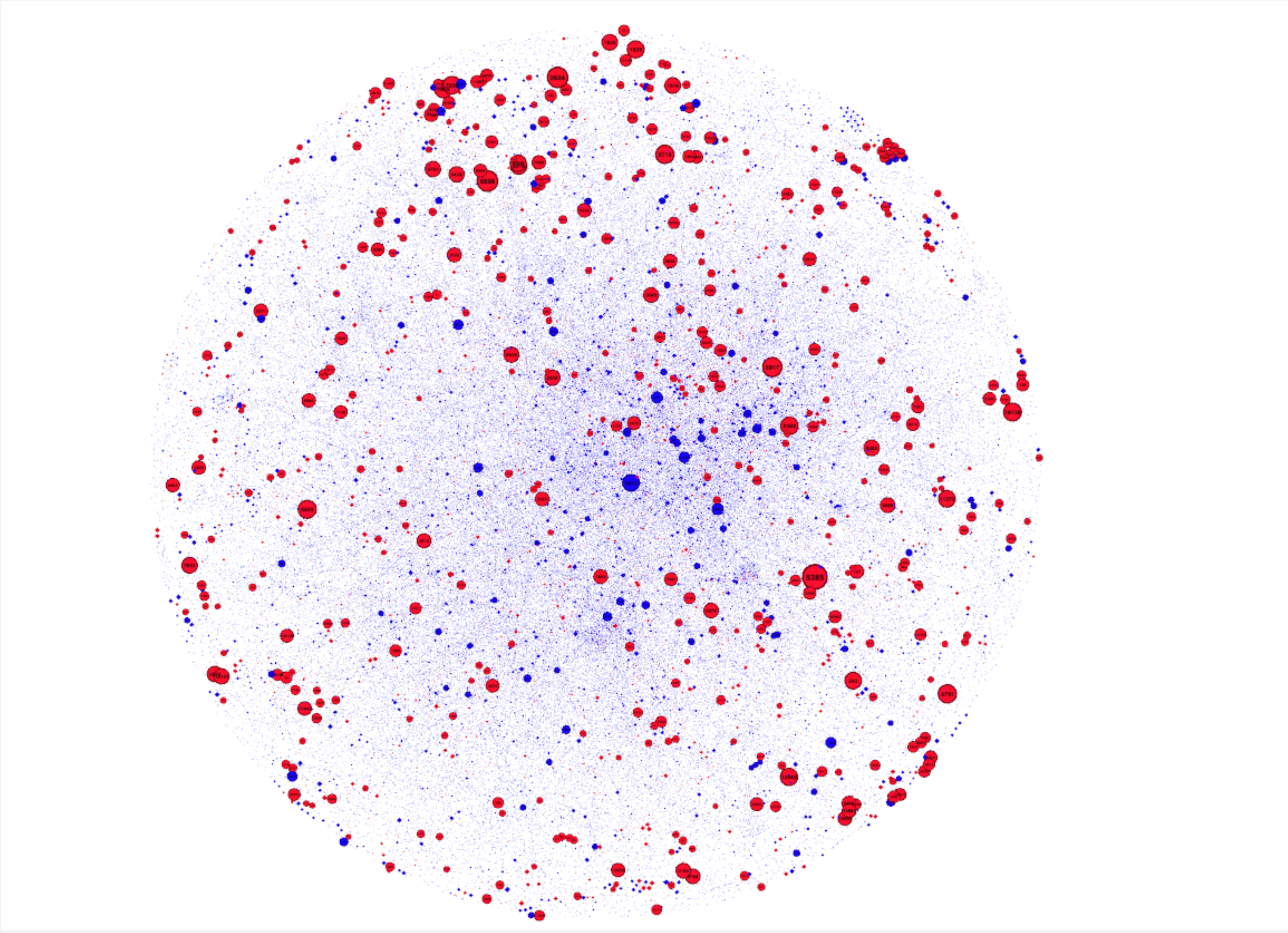}}
	}
	\subfigure[Jump vector for {\sensitive}]{
		{\includegraphics[width = 0.22\textwidth]{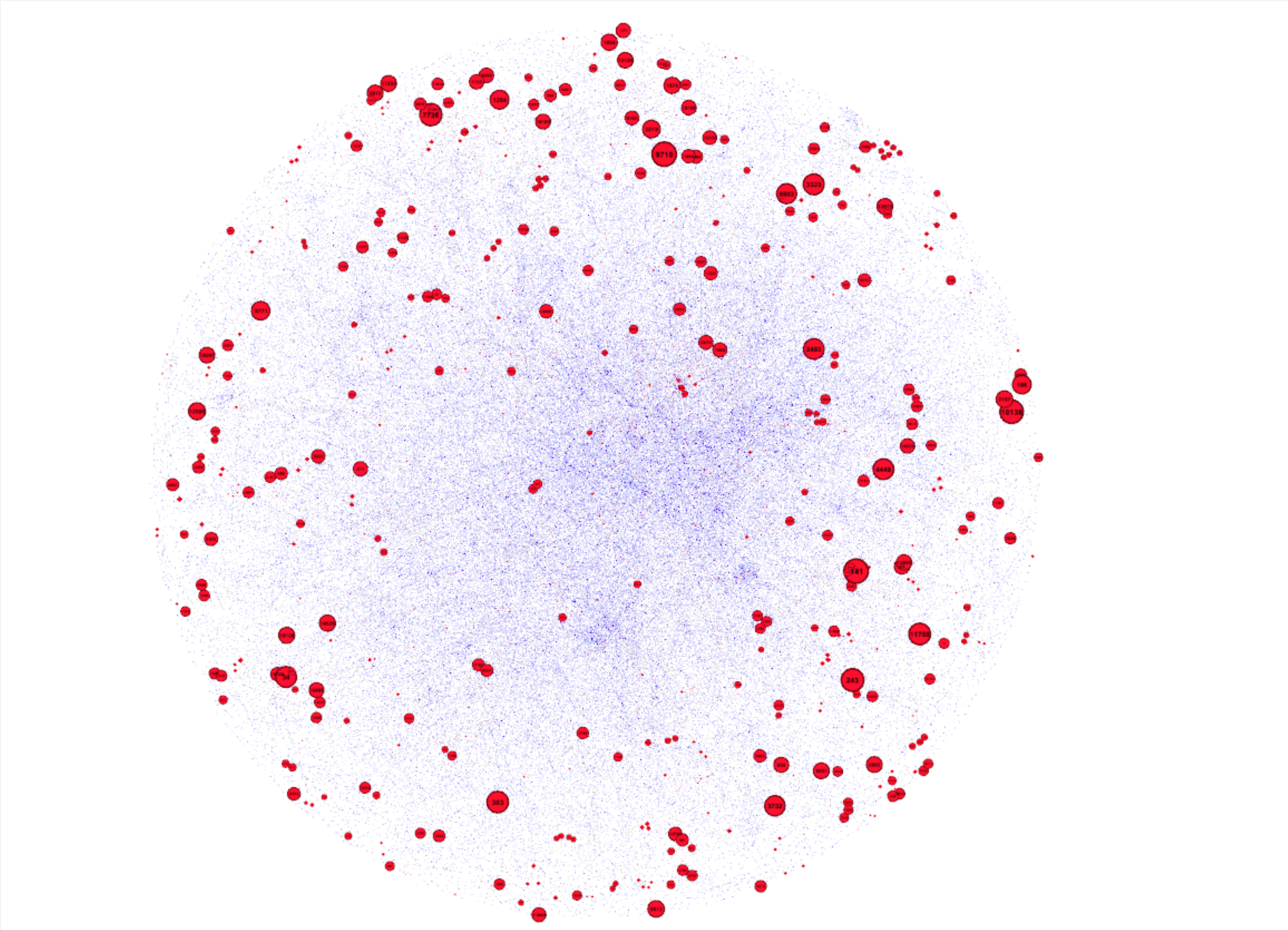}}
	}
	\caption{Visualization of the {\sc dblp} dataset.}
	\label{fig:vis-dblp}
\end{figure}

\noindent \textbf{Visualization.} 
In Figures \ref{fig:vis-books} and  \ref{fig:vis-dblp}, we visualize the results of the algorithms for the {\sc books}  and the {\sc dblp} dataset respectively, for $\phi$ = 0.5. Red nodes are colored red, and blue nodes are colored blue. Their size depends on the value of the quantity we visualize. We visualize the pagerank values for the original Pagerank, {\sensitive} and {\neighborlocal} algorithms, and the jump vector probabilities for  {\sensitive}. 

For  {\sc books}, where the original red pagerank is close to  $\phi$,
{\sensitive} is very similar to the original Pagerank algorithm.                                                                                                                                                                                                                                                                                                                     
{\sc books} is homophilic and the jump vector assigns rather uniform weights                                                                                                                                                                                                                                                                                                         
to almost all nodes.                                                                                                                                                                                                                                                                                                                                                                 
On the other hand, {\neighborlocal} promotes heavily nodes connecting the two opposite groups, i.e., nodes that are minorities in their neighborhoods. 
We observe a different result in {\sc dblp}, where $\phi$ is much larger than
the original red pagerank.
{\neighborlocal} distributes weights broadly in the red community, while {\sensitive} is much more aggressive. This is especially evident in the jump vector which promotes a few nodes in the periphery.

\begin{table}[ht]
	\centering
	\caption{Top-10 authors with $\phi$ = 0.3; the number in parenthesis is the position of the author in the original {\pagerank} ($\OPR$)  (female authors in bold).
	}
\footnotesize{
	\begin{tabular}{l l l}
		\toprule
		$\OPR$ &      \sensitive  &             \neighborlocal \\
		\midrule
		C. Faloutsos &          C. Faloutsos (1) &     C. Faloutsos (1)  \\
		G. Weikum &         G. Weikum  (2) &         G. Weikum (2) \\ 
		P. S. Yu &                  P. S. Yu  (3) &         \textbf{J. Widom}  (38) \\
		M. Stonebraker &    M. Stonebraker  (4) &    M. Stonebraker (4) \\ 
		M. J. Franklin &        M. J. Franklin (5)  & P. S. Yu  (3)\\
		H. Garcia-Molina &     H. Garcia-Molina  (6) &    \textbf{S. T. Dumais}  (28)  \\
		D. Kossmann &     D. Kossmann   (7) &      \textbf{M. Lalmas}   (27)  \\
		W. Lehner &   \textbf{E. A. Rundensteiner}   (22) &     P. Serdyukov   (17)  \\
		M. J. Carey  &     R. Agrawal   (11) &          \textbf{E. Bertino}   (25) \\
		M. de Rijke &    W. Lehner                  (8) &  \textbf{E. A. Rundensteiner}   (22) \\
		\bottomrule
	\end{tabular}
\label{table:case-study}
}
\end{table}

\noindent 
\textbf{Anecdotal Examples.} 
We will use the {\sc dblp} dataset for a qualitative evaluation of the results of the algorithms. Recall that for  this dataset, women are the minority with $r$ = 0.17, and the original red pagerank is equal to 0.16. 

To achieve a  fairer representation of women, we apply the {\neighborlocal} and {\sensitive} algorithms with $\phi$ = 0.3. In Table \ref{table:case-study}, we present the first 10 authors for each algorithm.
In the original Pagerank algorithm $\OPR$, there is no female author in the top-10 (the first one appears in position 22). {\sensitive} achieves fairness but also  minimizes utility loss, so the result is fair but also close to that of  $\OPR$.  {\neighborlocal} asks that all authors  
have $\phi$-fair inter-gender collaborations  resulting
 in a large number of female authors appearing in the top-10 positions.

\begin{table}[ht]
	\centering
	\caption{Top-3 female authors by conference, $\phi$ = 0.3.}
	\footnotesize {
		\begin{tabular}{lll}
			\toprule
			&SIGIR&SIGMOD\\
			\midrule
			\multirow{3}{*}{\sensitive}&
			Mounia Lalmas  & Elke A. Rundensteiner \\
			&Susan T. Dumais  & Elisa Bertino \\
			&Juliana Freire  & Tova Milo \\	
			\midrule
			\multirow{3}{*}{\neighborlocal}&
			Susan T. Dumais & Jennifer Widom \\
			& Mounia Lalmas  & Elke A. Rundensteiner \\
			& Emine Yilmaz  & Fatma Ozcan \\
			\bottomrule
		\end{tabular}
	}
	\label{table:case-study-conf}
\end{table}

We also use {\sc dblp} to study the targeted fair Pagerank algorithms.
In this case, we want to enforce fair behavior towards authors in specific conferences. We consider two such conferences, SIGIR and SIGMOD, and  select $S$ to include authors of each one of them.  In Table \ref{table:case-study-conf}, we show the top-3 women authors for each conference according to our algorithms. We observe that the algorithms produce different results depending on the conference, each time promoting women that are authorities in their respective fields, such as, Suzan Dumais when $S$ is SIGIR, and Jennifer Widom when $S$ is SIGMOD.

%% file: sec7.tex
\section{Related Work}
\label{sec:related-work}

\noindent \textbf{Algorithmic  fairness.} Recently, there has been increasing interest in algorithmic fairness, especially in the context of machine learning.
Fairness is regarded as the lack of discrimination on the basis of some protective attribute.
Various definition of  fairness having proposed especially for classification \cite{fairness-awarness,fairness-study,fairness-measures,bias-fairness-survey}. We use a group-fairness definition, based on parity.
Approaches to handing fairness can be classified as \textit{pre-processing}, that modify the input data, \textit{in-processing}, that modify the algorithm and \textit{post-processing} ones, that modify the output. We are mostly interested in in-processing techniques.

There is also prior work on fairness in ranking~\cite{fai*r,Julia17,Biega18,pair-wise}. 
All of these works consider ranking as an ordered list of items, and use different rules for defining and enforcing fairness that consider different prefixes of the ranking~\cite{fai*r,Julia17}, pair-wise orderings~\cite{pair-wise}, or exposure and presentation bias~\cite{exposure-ranking,Biega18}. 


Our goal in this paper is not to propose a new definition of ranking fairness, but rather to initiate a study of fairness in link analysis.
A distinguishing aspect of our approach is that we
take into account the actual Pagerank weights of the nodes, not just their ranking. 
%
Furthermore, our focus in this paper, is to design in-processing algorithms that incorporate fairness in the inner working of the Pagerank algorithm. We present a post-processing approach as a means to estimate a lower bound on the utility loss. None of the previous approaches considers ranking in networks, so the proposed approaches are novel. 


\noindent \textbf{Fairness in networks.} 
There has been some  recent work on network fairness in the context of graph embeddings \cite{fair-embedding,filter-bubbles,fair-walk}.
The work in \cite{fair-embedding} follows an in-processing approach 
that extends the learning function with regulatory fairness enforcing terms, while the work in  \cite{filter-bubbles} follows a post-processing approach so as to promote link recommendations to nodes belonging
to specific groups. Both works are not related to our approach.
The work in \cite{fair-walk} extends the node2vec graph embedding method by modifying the
random walks used in node2vec with fair walks, where nodes are partitioned into groups and each group is given the same probability of being selected when a node makes a transition. The random walk 
introduced in \cite{fair-walk} has some similarity with the
random walk interpretation of {\neighborlocal}. It would
be interesting to see, whether our extended residual-based algorithms could be utilized also in the context of graph embeddings, besides its use in link analysis.

There are also previous studies on the effect of
homophily, preferential attachment and imbalances in group sizes.
It was shown that the combination of these three factors leads to 
uneven degree distributions between groups \cite{glass-ceiling}.
Recent work shows that this phenomenon is exaggerated by many link recommendation algorithms \cite{glass-ceiling-recommend}. Evidence of inequality between degree distribution of minorities and majorities was also found in many real networks \cite{homophily-ranking}.
Our work extends this line of research by looking at Pagerank values instead of degrees.
Along this lines, recent work studies in depth how homophily and size imbalance can affect the visibility
that different groups get in link recommendations, i.e, how often nodes in each group get recommended \cite{xyz}.
Very recent work also looks at graph mining algorithms in general from the perspective of individual fairness, where the goal is to produce a similar output for similar nodes \cite{inform}.

Finally, there is previous work on diversity in network ranking. 
The goal is to find important nodes that also maximally cover the nodes in the network~\cite{grasshoper,divrank}.
Our problem is fundamentally different, since we look for scoring functions that follow a parity constraint.


%% file: sec8.tex
\section{Conclusions}
\label{sec:conclusions}
In this paper, we initiate a study of fairness for Pagerank.
We provide definitions of fairness, and we propose two approaches for achieving fairness: one that modifies the jump vector, and one that imposes a fair behavior per node. We prove that the latter is equivalent to a stronger notion of fairness that also guarantees personalized Pagerank fairness. We also consider the problem of attaining fairness while minimizing the utility loss of Pagerank. Our experiments demonstrate the behavior of our different algorithms.

There are many direction for future work. First, we would like to study the role of
$\gamma$, i.e., the jump probability. Then, it would be interesting to explore other notions of Pagerank fairness, besides $\phi$-fairness, for instance ones based on
rank-aware fairness \cite{edbt-tut}. Furthermore, we plan to explore further applications of the theory behind our fair Pagerank algorithms to derive novel notions of fair random walks.